%% file: popl24.tex
\documentclass[acmsmall,screen,nonacm]{acmart}
\usepackage{amsmath,amsfonts,amsthm,mathtools}
\usepackage{algorithmic}
\usepackage{listings}

\definecolor{mygreen}{rgb}{0,0.4,0}
\definecolor{myblue}{rgb}{0,0,0.7}
\definecolor{myred}{rgb}{0.7,0,0}
\definecolor{mygray}{rgb}{0.5,0.5,0.5}
\definecolor{mymauve}{rgb}{0.58,0,0.82}

\usepackage{soul}
\usepackage{wrapfig}

\newcommand*{\mlstinlin}[1]{\mbox{\lstinline|#1|}}
\usepackage[export]{adjustbox}
\usepackage{textcomp}
\usepackage{xcolor}
\usepackage{proof}
\usepackage{stmaryrd}
\usepackage{todonotes}
\usepackage{tikz}
\usepackage{tikz-cd}
\usepackage{cleveref}
\usepackage[caption=false]{subfig}
\DeclareMathOperator*{\colim}{colim}

\newtheorem{theorem}{Theorem}
\newtheorem{proposition}[theorem]{Proposition}
\newtheorem{definition}[theorem]{Definition}

\newtheorem{example}[theorem]{Example}
\newtheorem{corollary}[theorem]{Corollary}
\usepackage{microtype}
\DeclareMathOperator{\Aut}{Aut}
\usepackage{mathrsfs}
\usepackage{newunicodechar}
\newunicodechar{：}{\colon}
\newunicodechar{∀}{\forall}
\newunicodechar{∃}{\exists}
\newunicodechar{α}{\alpha}
\newunicodechar{β}{\beta}
\newunicodechar{γ}{\gamma}
\newunicodechar{δ}{\delta}
\newunicodechar{ε}{\epsilon}
\newunicodechar{ζ}{\zeta}
\newunicodechar{η}{\eta}
\newunicodechar{θ}{\theta}
\newunicodechar{ι}{\iota}
\newunicodechar{κ}{\kappa}
\newunicodechar{λ}{\lambda}
\newunicodechar{μ}{\mu}
\newunicodechar{ν}{\nu}
\newunicodechar{ξ}{\xi}
\newunicodechar{ρ}{\rho}
\newunicodechar{σ}{\sigma}
\newunicodechar{τ}{\tau}
\newunicodechar{υ}{\upsilon}
\newunicodechar{φ}{\phi}
\newunicodechar{χ}{\chi}
\newunicodechar{ψ}{\psi}
\newunicodechar{ω}{\omega}
\newunicodechar{Σ}{\Sigma}
\newunicodechar{π}{\pi}
\newunicodechar{Π}{\Pi}
\newunicodechar{Γ}{\Gamma}
\newunicodechar{Δ}{\Delta}
\newunicodechar{⊢}{\vdash}
\newunicodechar{∈}{\in}
\newunicodechar{∉}{\notin}
\newunicodechar{∋}{\ni}
\newunicodechar{∌}{\not\ni}
\newunicodechar{∅}{\emptyset}
\newunicodechar{∪}{\cup}
\newunicodechar{∩}{\cap}
\newunicodechar{∧}{\wedge}
\newunicodechar{∨}{\vee}
\newunicodechar{⊥}{\bot}
\newunicodechar{⊤}{\top}
\newunicodechar{≤}{\leq}
\newunicodechar{≥}{\geq}
\newunicodechar{≠}{\neq}
\newunicodechar{≡}{\equiv}
\newunicodechar{⊆}{\subseteq}
\newunicodechar{⊇}{\supseteq}
\newunicodechar{⊂}{\subset}
\newunicodechar{⊃}{\supset}
\newunicodechar{⟶}{\rightarrow}
\newunicodechar{→}{\to}
\newunicodechar{↪}{\hookrightarrow}
\newunicodechar{⟵}{\leftarrow}
\newunicodechar{↦}{\mapsto}
\newunicodechar{⟹}{\Rightarrow}
\newunicodechar{⟺}{\Leftrightarrow}
\newunicodechar{⊗}{\otimes}
\newunicodechar{×}{\times}
\newunicodechar{∞}{\infty}
\newunicodechar{…}{\ldots}
\newunicodechar{⋯}{\cdots}
\newunicodechar{⋮}{\vdots}
\newunicodechar{⋱}{\ddots}
\newunicodechar{⋰}{\ddots}
\newunicodechar{∧}{\wedge}
\newunicodechar{∨}{\vee}
\newunicodechar{∩}{\cap}
\newunicodechar{∪}{\cup}
\newunicodechar{¬}{\neg}
\newunicodechar{ℙ}{\mathbb{P}}
\newunicodechar{ℕ}{\mathbb{N}}
\newunicodechar{ℝ}{\mathbb{R}}
\newunicodechar{𝕍}{\mathbb{V}}
\newunicodechar{≝}{\defeq}
\newunicodechar{𝒞}{\mathscr{C}}
\newunicodechar{𝒟}{\mathscr{D}}

\newcommand{\hide}[1]{}

\newcommand{\tunit}{\mathsf{unit}}
\newcommand{\treal}{\mathsf{real}}
\newcommand{\tbool}{\mathsf{bool}}
\newcommand{\tvertex}{\mathsf{vertex}}

\newcommand{\ttrue}{\mathsf{true}}
\newcommand{\tfalse}{\mathsf{false}}

\newcommand{\inj}{\mathsf{in}}
\newcommand{\case}[5]{\mathsf{case}\,#1\, \mathsf{of}\,\{\inj_1(#2)\Rightarrow #3;\inj_2(#4)\Rightarrow #5\}}
\newcommand{\casezero}[1]{\mathsf{case}\,#1\,\mathsf{of}\,\{\}}

\newcommand{\tnormal}{\mathsf{normal}}
\newcommand{\tbernoulli}{\mathsf{bernoulli}}
\newcommand{\tnew}{\mathsf{new}}
\newcommand{\tedge}{\mathsf{edge}}
\newcommand{\letin}[2]{\mathsf{let}\,#1\,=\,#2\,\mathsf{in}\,}
\newcommand{\letinsqueeze}[2]{\mathsf{let}\,#1{=}#2\,\mathsf{in}\,}
\newcommand{\s}{\,|\,}
\newcommand{\sem}[1]{\llbracket #1\rrbracket}

\newcommand{\RR}{\mathbb{R}}
\newcommand{\NN}{\mathbb{N}}
\newcommand{\VV}{\mathbb{V}}
\newcommand{\AutRado}{\Aut(\mathrm{Rado})}
\newcommand{\id}{\mathrm{id}}
\newcommand{\RadoNom}{\mathbf{RadoNom}}

\newcommand{\inv}{^{\raisebox{0.03ex}{{\textrm{-}}}1}}
\newcommand{\dd}{\mathrm{d}}
\newcommand{\bind}{\mathrel{\scalebox{0.5}[1]{$>\!>=$}}}

\newcommand{\defeq}{\stackrel {\mathrm{def}}=}

\newcommand{\Set}{\mathbf{Set}}
\newcommand{\FinStoch}{\mathbf{FinStoch}}
\newcommand{\FinSet}{\mathbf{FinSet}}
\newcommand{\CatA}{\mathbf{C}}
\newcommand{\CatB}{\mathbf{D}}
\newcommand{\CatX}{\mathbf{X}}
\newcommand{\CatC}{\mathbf{A}}
\newcommand{\Meas}{\mathbf{Meas}}
\newcommand{\DistM}{\mathcal{D}}
\newcommand{\Giry}{\mathcal{G}}

\newcommand{\fv}{\mathsf{fv}}
\newcommand{\subm}{\mathsf{subm}}

\newcommand{\ie}{\textit{i.e. }}

\newcommand\restrict[1]{\raisebox{-.5ex}{$|$}_{#1}}

\newcommand{\ErdosRenyi}{Erd\H{o}s--R\'enyi}

\newcommand{\CatAp}{\CatA/\!_\bbase}
\newcommand{\Fam}{\mathsf{Fam}}
\newcommand{\bbase}{\Psi}

\setcopyright{none}

\begin{document}

\title{Probabilistic Programming Interfaces for Random Graphs: Markov Categories, Graphons, and Nominal Sets}

\author{Nate Ackerman}
\orcid{0000-0002-8059-7497}
\affiliation{%
  \institution{Harvard University}
  \country{USA}
}
\email{nate@aleph0.net}

\author{Cameron E. Freer}
\orcid{0000-0003-1791-6843}
\affiliation{%
  \institution{Massachusetts Institute of Technology}
  \country{USA}
}
\email{freer@mit.edu}

\author{Younesse Kaddar}
\orcid{0000-0001-7366-9889}
\affiliation{%
  \institution{University of Oxford}
  \country{UK}
}
\email{younesse.kaddar@chch.ox.ac.uk}

\author{Jacek Karwowski}
\orcid{0000-0002-8361-2912}
\affiliation{%
  \institution{University of Oxford}
  \country{UK}
}
\email{jacek.karwowski@cs.ox.ac.uk}

\author{Sean Moss}
\orcid{0009-0005-3930-5340}
\affiliation{%
  \institution{University of Birmingham}
  \country{UK}
}
\email{s.k.moss@bham.ac.uk}

\author{Daniel Roy}
\orcid{0000-0001-8930-0058}
\affiliation{%
  \institution{University of Toronto}
  \country{Canada}
}
\email{daniel.roy@utoronto.ca}

\author{Sam Staton}
\orcid{0000-0002-0141-8922}
\affiliation{%
  \institution{University of Oxford}
  \country{UK}
}
\email{sam.staton@cs.ox.ac.uk}

\author{Hongseok Yang}
\orcid{0000-0003-1502-2942}
\affiliation{%
  \institution{KAIST}
  \department{School of Computing}
  \country{South Korea}
}
\email{hongseok00@gmail.com}
\renewcommand{\shortauthors}{N.~Ackerman, C.~Freer, Y.~Kaddar, J.~Karwowski, S.~Moss, D.~Roy, S.~Staton, H.~Yang}

\begin{abstract}
  We study semantic models of probabilistic programming languages over
  graphs, and establish a connection to graphons from graph
  theory and combinatorics. We show that every well-behaved
  equational theory for our graph probabilistic programming language
  corresponds to a graphon, and conversely, every graphon arises in this way.
  
  We provide three constructions for showing that every graphon arises
  from an equational theory. The first is an abstract construction, using Markov
  categories and monoidal indeterminates. 
  The second and third are more concrete. 
  The second is in terms of traditional measure theoretic probability,
  which covers `black-and-white' graphons.
  The third is in terms of probability monads on the nominal sets of
  Gabbay and Pitts. Specifically, we use a variation of nominal sets
  induced by the theory of graphs, which covers Erdős-Rényi graphons. 
  In this way, we build new models of graph probabilistic
  programming from graphons. 
\end{abstract}



\maketitle

\input{introduction.tex}

\input{programming_with_graphs.tex}

\input{equations_to_graphons.tex}
\input{graphons_to_equations_syntactic.tex}

\input{random_free.tex}
\input{rado_nominal_sets.tex}
\input{graphons_to_equations.tex}

\input{related_work.tex}

\begin{acks}
  This material is based on work supported by a Royal Society University Research Fellowship, ERC Project BLAST and AFOSR Award No.~FA9550–21–1–003. This work was supported in part by CoCoSys, one of seven centers in JUMP 2.0, a Semiconductor Research Corporation (SRC) program sponsored by DARPA. HY was supported by the National Research Foundation of Korea (NRF) grant funded by the Korea government (MSIT) (No. RS-2023-00279680), and also by the Engineering Research Center Program through the National Research Foundation of Korea (NRF) grant funded by the Korean government (MSIT) (No. NRF-2018R1A5A1059921).

It has been helpful to discuss developments in this work with many people over the years, and we are also grateful to our anonymous reviewers for helpful suggestions. 
\end{acks}

\bibliographystyle{ACM-Reference-Format}
\interlinepenalty=10000
\enlargethispage{\baselineskip}
\bibliography{bibliography}

\end{document}

%% file: introduction.tex
\section{Introduction}
\label{sec:intro}

This paper is about the semantic structures underlying probabilistic
programming with random graphs. Random graphs have applications in
statistical modelling across biology, chemistry, epidemiology, and so
on, as well as theoretical interest in graph theory and combinatorics (e.g.~\cite{graphs-handbook}). Probabilistic programming, i.e.~programming for statistical
modelling~\cite{probprog-intro}, is useful for building the statistical
models for the applications. Moreover, as we show
(Theorem~\ref{thm:eq-theory-to-graphon} and Corollary~\ref{corollary:graphon-markov}),
the semantic aspects of programming languages for random graphs
correspond to graphons~\cite{MR3012035}, a core structure in graph
theory and combinatorics.

To set the scene more precisely, we recall the setting of
probabilistic programming with real-valued distributions, 
and contrast it with the setting with graphs.
Many probabilistic programming languages provide a type of
real numbers ($\treal$) and distributions such as the normal
distribution
\begin{equation}\label{eqn:tnormal}
  \tnormal :\treal\ast\treal\to \treal
\end{equation}
together with arithmetic operations such as
\begin{equation}\label{eqn:add}
  (+) :\treal\ast\treal\to \treal\text.
\end{equation}
Even if we encounter an unfamiliar distribution over $(\treal)$ in a library, we have
a rough idea of how to explain what it could be, in terms of probability
densities and measures.

In this paper, we consider the setting of probabilistic programming with graphs, where
the probabilistic programming language or library provides a type $(\tvertex)$ and some distribution 
\begin{equation}\label{eqn:intro-new}
  \tnew : \tunit \to \tvertex
\end{equation}
together with a test
\begin{equation}\label{eqn:intro-edge}
  \tedge:\tvertex\ast\tvertex\to \tbool\text.
\end{equation}
Our goal is to analyze the interface $(\tvertex,\tnew,\tedge)$ for graphs semantically,
and answer, for instance, what they could be and what they could do. We give one example analysis in Section~\ref{sec:intro:graphs} first, and the general one later in Theorem~\ref{thm:eq-theory-to-graphon} and Corollary~\ref{corollary:graphon-markov}, which says that to give an implementation of
$(\tvertex,\tnew,\tedge)$, satisfying the laws of probabilistic
programming, is to give a graphon. In doing so, we connect 
the theory of probabilistic programming with graph theory and
combinatorics.

Probabilistic programming is generally used for statistical inference, in which we describe a generative model by
writing a program using primitives such as \eqref{eqn:tnormal}--\eqref{eqn:intro-edge} above, and then infer a distribution on certain parameters, given particular observed data.  This paper is focused on the generative model aspect, and not inference (although for simple examples, generic inference methods apply immediately, see~\S\ref{sec:intro:practice}). 

\subsection{Example of an Implementation of a Random Graph: Geometric Random Graphs}
\label{sec:intro:graphs}

To illustrate the interface $(\tvertex,\tnew,\tedge)$ of
\eqref{eqn:intro-new}--\eqref{eqn:intro-edge},
we consider for illustration a random geometric graph (e.g.~\cite{penrose-rgg,https://doi.org/10.1002/rsa.20633}) where the vertices are
points on the surface of the unit sphere, chosen uniformly at random, and where there is an
edge between two vertices if the angle between them is less than
some fixed~$\theta$.
This random graph might be used, for instance, to model the connections between people on the earth.

For example, a simple statistical inference problem might start from the observed connectivity in Figure~\ref{fig:graph}(a). 
We might ask for the distribution on $\theta$ given that this graph arose from the spherical random geometric graph. One sample from this posterior distribution on random geometric graphs with $\theta=\pi/3$ is shown in Figure~\ref{fig:graph}(b). Another, unconditioned sample from the random geometric graph with $\theta=\pi/6$ is shown in Figure~\ref{fig:graph}(c).  

We can regard this example as an implementation of the interface
$(\tvertex,\tnew,\tedge)$ as follows:
we implement $(\tvertex)$ as the surface of the sphere (e.g. implemented
as Euclidean coordinates). 
\begin{itemize}
\item $\tnew():\tvertex$ randomly picks a new vertex
  as a point on the sphere uniformly at random.
  Figure~\ref{fig:graph}(c) shows the progress after calling $\tnew()$ 15
  times.
\item $\tedge:\tvertex\ast \tvertex \to \tbool$ checks whether there
  is an edge between two vertices; this amounts to checking whether the angle between two points is less than $\theta$. 
  \end{itemize}
\begin{figure}\centering
(a)\includegraphics[scale=1,angle=110,valign=t]{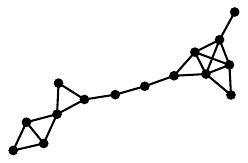} \quad (b)\includegraphics[scale=.5,valign=t]{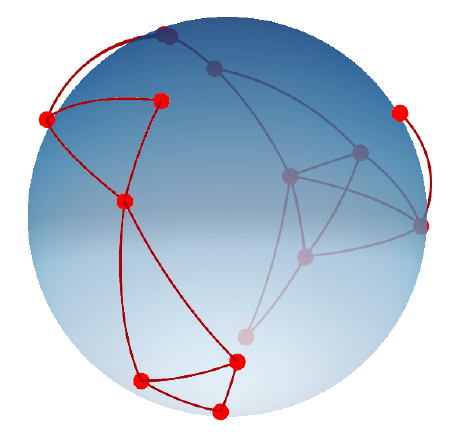} \quad
(c)\includegraphics[scale=.5,valign=t]{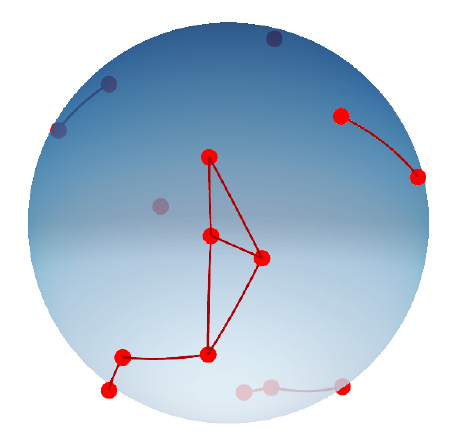}
  \vspace{-1cm}
\caption{(a)~A graph; (b)~an inferred geometric realization of it ($θ \approx π/3$); (c)~a generated sample for $θ = π/6$.\label{fig:graph}}
\end{figure}
For a simple example, we can write a program over the interface to calculate the probability of three
  random vertices forming a triangle: the program
\begin{equation}\label{eqn:triangle}\begin{aligned}&\letin a{\tnew()}{\letin b{\tnew()}{\letin c {\tnew()}
    {}}}\\&\qquad\tedge(a,b) \,\&\, \tedge(b,c) \,\&\,\tedge (a,c):\tbool\end{aligned}
\end{equation}
randomly returns true or false; the probability of true is the
probability of a triangle. 

This implementation using the sphere is only one way to implement
$(\tvertex,\tnew,\tedge)$. There are
implementations using higher-dimensional spheres, or other geometric
objects. We can also consider random equivalence relations as
graphs, i.e.~disjoint unions of complete graphs, or random bipartite
graphs, which are triangle-free. We can consider the
\ErdosRenyi\ random graph, where the chance of an edge between two
vertices is independent of the other edges, and has a fixed
probability.
These are all different implementations of the same abstract
interface, $(\tvertex,\tnew,\tedge)$,
and programs such as \eqref{eqn:triangle} make sense for all of them. 
The point of this paper is to characterize all these
implementations, as graphons. 

\subsection{Implementations Regarded as Equational Theories}\label{sec:intro:eqthys}
The key method of this paper is to treat implementations of the
interface ($\tvertex,\tnew,\tedge$) extensionally, as equational
theories. That is, rather than looking at specific implementation
details, we look at the equations between programs that a user of the
implementation would rely on. (This is analogous to the idea in model
theory of studying first-order theories rather than specific models;
similar ideas arise in the algebraic theory of computational
effects~\cite{pp-notions}.) 
For example, if an implementation always provides a
bipartite random graph, we have the equation 
\[\text{Program~\eqref{eqn:triangle}} \ \ \equiv\ \ \tfalse \qquad\text{between programs,}\]
because a triangle is never generated. This equation does not hold in the example of Figure~\ref{fig:graph}(b--c), since
triangles are possible. 

We focus on a class of equational theories that are well behaved, as
follows.
First, we suppose that they contain basic laws for
probabilistic programming
(eqns.~(\ref{eq:let-assoc})\,--\,(\ref{cd:let_val}), \S\ref{sec:markov-cats}). This basic structure
already appears
broadly in
different guises, including in Moggi's monadic metalanguage~\cite{moggi-computational-lambda}, in linear
logic~\cite{prob-cbpv}, and in synthetic
probability theory~\cite{fritz}. Second, we suppose that the equational theories 
are equipped with a `Bernoulli base', which means that although we do not specify an implementation for the type $(\tvertex)$, each closed program of
type $(\tbool)$ is equated with some ordinary Bernoulli distribution, in such a way as to satisfy the classical laws of traditional finite probability
theory (\S~\ref{sec:bernoulli-base}).
Finally, we suppose that the edge relation is symmetric (the graphs are
undirected) and that it doesn't change when the same question is asked
multiple times (`deterministic'), e.g.
\begin{equation}\letin a{\tnew()}{\letin
                                b{\tnew()}{\tedge(a,b) \,\&\, \neg
                                \tedge(a,b)}}\ \  \ \equiv\ \ \ \ \tfalse \text. \label{eqn:intro:det}\end{equation}

A \emph{graphon} is a symmetric measurable function ${[0,1]^2\to
  [0,1]}$. We show
that every equational theory for the interface
$(\tvertex,\tnew,\tedge)$ gives rise to a graphon (Theorem~\ref{thm:eq-theory-to-graphon}), and conversely that
every graphon arises in this way
(Corollary~\ref{corollary:graphon-markov}).

We emphasize that this abstract treatment of implementations, in terms
of equational theories, is very open-ended, and permits a diverse range of implementation
methods. Indeed, we show in Section~\ref{sec:randomfree} that any
implementation using traditional measure-theoretic methods will
only produce black-and-white graphons, so this abstract treatment is crucial. 

\subsection{From Equational Theories to Graphons}
\label{sec:intro:equations-to-graphons}

In Section~\ref{sec:prog-eq-to-graphons}, we show how an equational
theory over programs in the interface $(\tvertex,\tnew,\tedge)$ gives rise to  a graphon.
The key first step is that graphons (modulo equivalence) can be
characterized in terms of sequences of finite random graphs that
satisfy three conditions: exchangeability, consistency, and locality.

To define a graphon, we show how to define programs that describe finite random graphs, by using
$\tnew$ and $\tedge$ to build boolean-valued $n\times n$ adjacency
matrices, for all $n$ (shown in~\eqref{eqn:programgraph}).
Assuming that the equational theory of programs is Bernoulli-based,
these programs can be interpreted as probability distributions on the
finite spaces of adjacency matrices which, we show, are finite random graphs.  

It remains to show that the induced sequence of random graphs
satisfies the three conditions for graphons (exchangeability,
consistency, and locality). These can be formulated as equational
properties, and so they can be verified by using the equational
reasoning in the equational theory. This is Theorem~\ref{thm:eq-theory-to-graphon}.
A key part of the proof is the observation that exchangeability for graphons
connects to commutativity of $\mathsf{let}$~\eqref{eq:let-comm}: we can permute the order in which
vertices are instantiated without changing the distributions.

\subsection{From Graphons to Equational Theories}
\label{sec:intro:graphons-to-equations}
We also show the converse: every graphon arises from a good
equational theory for the interface $(\tvertex,\tnew,\tedge)$.
We look at this from three angles: first, we prove this in the general
case using an abstract method, and then, we use concrete methods for
two special cases.

Fixing a graphon, we build an equational theory by following a
categorical viewpoint. A good equational theory for probabilistic
programming amounts to a `distributive Markov category', which is a
monoidal category with coproducts that is well-suited to probability
(\S\ref{sec:markov-cats} and \cite{fritz}). The idea that distributive categories
are a good way to analyze abstract interfaces goes back at least
to~\cite{walters}, which used distributive categories to study
interfaces for stacks and storage. 
We can thus use now-standard abstract methods for building monoidal and
distributive categories to build an equational theory for the
programming language.

We proceed in two steps. We first use methods such as~\cite{hermida-tennent,hu-tholen}
to build an abstract distributive
Markov category that supports the interface
$(\tvertex,\tnew,\tedge)$ in a generic way. This equational theory is
generic and not Bernoulli-based: although it satisfies the equational laws of
probabilistic programming, there is no given connection to traditional probability.
The second step is to  show that (a)~it is possible to quotient this generic category
to get Bernoulli-based equational theories; (b)~the choices of
quotient are actually in bijective correspondence with graphons.
Thus, we can build an equational theory from which any given graphon arises,
via~\eqref{eqn:programgraph}: this is Corollary~\ref{corollary:graphon-markov}.
(The framework of Bernoulli-based Markov categories is new here, and the techniques of~\cite{hermida-tennent,hu-tholen}  have not previously been applied in categorical probability, so
a challenge for future work is to investigate these ideas in other aspects of categorical probability.)

Although this is a general method,
it is an abstract method involving quotient constructions. 
The ideal form of denotational semantics is to explain what programs are by regarding them as
functions between certain kinds of spaces. Although Corollary~\ref{corollary:graphon-markov} demonstrates
that every graphon arises from an equational theory, the type
$(\tvertex)$ is interpreted as an object of an abstract category, and
programs are equivalence classes of abstract morphisms.
In the remainder of the paper, we give two situations where we can
interpret $(\tvertex)$ as a genuine concrete space, and programs are
functions or distributions on spaces. Such an interpretation
immediately yields an equational theory, where two programs are
equal if they have the same interpretation.

\begin{itemize}
\item \emph{Section~\ref{sec:randomfree}}: For `black-and-white graphons', we present 
  measure-theoretic models of the interface, based on a standard
  measure-theoretic interpretation of probabilistic programming
  (e.g.~\cite{Kozen81}).
  We interpret $(\tvertex)$ as a measurable space, and $(\tnew)$ as a
  probability measure on it, and $(\tedge)$ in terms of a measurable predicate. 
  Then, the composition of programs is defined in terms of probability
  kernels and Lebesgue integration.
  This kind of model exactly captures the black-and-white graphons (Prop.~\ref{prop:all-bw}).
\item \emph{Section~\ref{sec:ER-Rado}}: For `\ErdosRenyi' graphons, which are
  constantly gray, and not
  black-and-white, we present a model based on Rado-nominal sets (\S\ref{sec:rado-def}).
  These are a variant of nominal sets (\cite{gp-nominal,pitts-nom-book}) 
  where the atoms are vertices of
  the Rado graph (following~\cite{lmcs:1157}). We consider a new notion of `internal
  probability measure' in this setting, and use this to give a 
  compositional semantics that gives rise to the \ErdosRenyi\
  graphons (Corollary~\ref{cor:e-r}).
\end{itemize}
Together, these more concrete sections then provide further intuition for the
correspondence between equational theories and graphons. 

\subsection{Connection to Practice}\label{sec:intro:practice}
We conclude this introduction with remarks on the connection to practical modelling. 
In practice, the graph interface might form part of a generative model, on which we perform inference.
The structure is clearest in a typed language, and one example is the LazyPPL library for Haskell~\cite{lazyppl}.
(Similar examples are implemented in~\cite[Ch.~12]{probmods1},
 albeit untyped.)
Then our interface is captured by a Haskell type class:\footnote{See \url{https://lazyppl-team.github.io/GraphDemo.html} for full details in literate Haskell.}
\begin{lstlisting}
class RandomGraph p vertex | p -> vertex where
  new :: p -> Prob vertex
  edge :: p -> vertex -> vertex -> Bool
  \end{lstlisting}
Here \lstinline|p| is a parameter type, and we write \lstinline|Prob| for a probability monad. 
A spherical implementation of the interface (following \S\ref{sec:intro:graphs}) is parameterized by the dimension~$d$ and the distance $\theta$, as follows:
\begin{lstlisting}
data SphGrph = SG Int Double -- parameters for a sphere graph
data SphVertex = SV [Double] -- vertices are Euclidean coordinates
instance RandomGraph SphGrph SphVertex where
  new :: SphGrph -> Prob SphVertex
  new (SG d theta) = ... -- sample a random unit d-vector uniformly 
  edge :: SphGrph -> SphVertex -> SphVertex -> Bool
  edge (SG d theta) v w = ... -- check whether arccos(v.w) < theta
\end{lstlisting}
We can use this as a building block for more complex models.
For a simple example, we generated Figure~\ref{fig:graph}(b) by using the generic Metropolis-Hastings inference of the LazyPPL library to infer $\theta$ given a particular graph (Fig.~1(a)). We have also implemented other random graphs; our implementation of the \ErdosRenyi\ graph uses stochastic memoization~\cite{roy2008stochastic,kaddar-staton-stoch-mem}.

\subsubsection*{Summary and context}
As we have discussed, our main result is that equational theories for the programming interface (\S\ref{sec:intro:graphs})
give rise to graphons (\S\ref{sec:intro:equations-to-graphons}) and every graphon arises in this way~(\S\ref{sec:intro:graphons-to-equations}).

These results open up new ways to study random graphs, by using programming semantics. On the other hand, our results here put the abstractions of practical probabilistic programming on a solid theoretical foundation (see also~\S\ref{sec:relatedwork}).

%% file: programming_with_graphs.tex
\newcommand{\CatADet}{\CatA_{\mathsf{det}}}
\newcommand{\Kl}{\mathrm{Kl}}
\newcommand{\FP}{\mathrm{FP}}
\newcommand{\op}{^\mathrm{op}}

\section{Programming interfaces for random graphs: equational theories
and Markov categories}
In Section~\ref{sec:intro:graphs}, 
we considered probabilistic
programming over a graph interface. To make this formal, we now recall
syntax, types, and equational reasoning for simple probabilistic
programming languages.
We begin with a general syntax (\S\ref{sec:generic-ppl}), which can
accommodate various interfaces in the form of type and term
constants, including the interface for graphs (Ex.~\ref{ex:interfaces}(3)). 

We study different instantiations of the probabilistic programming
language in terms of the equational theories that they satisfy. We
consider two equivalent ways of understanding equational theories: as
distributive Markov categories (\S\ref{sec:markov-cats}) and in terms
of affine monads (\S\ref{sec:affine-monad}). Markov categories are a
  categorical formulation of probability theory (e.g.~\cite{fritz}), and affine
  monads arise in the categorical analysis of probability (e.g.~\cite{fgpt-weakly-affine,jacobs-commutative-effectus,kock:commutative-monads-as-a-theory-of-distributions})
  as well as in the semantics for probabilistic programming
  (e.g.~\cite{lazyppl,dsp-layer,amorim-markov}). 
We make a connection with traditional probability via the notion of
Bernoulli base (\S\ref{sec:bernoulli-base}).

Much of this section will be unsurprising to experts: the main purpose is to collect
definitions and results. The definition of \emph{distributive} Markov
category appears to be novel, and so we go over that definition and
correspondence with monads (Propositions~\ref{prop:distr-affine-monad} and~\ref{prop:markov-embed-kleisli}).
In Section~\ref{sec:quotients}, we give a construction for quotienting a distributive Markov category, which we will need in Section~\ref{sec:graphons-to-equational-theories}.
We include the result in the section because it may be of independent interest.

\subsection{Syntax for a Generic Probabilistic Programming Language}
\label{sec:generic-ppl}
Our generic probabilistic programming language is, very roughly, an idealized,
typed fragment of a typical language like Church~\cite{goodman2008church}. 
We start with a simple programming language (following
\cite{s-finite,prob-cbpv,dario-thesis} but also
\cite{moggi-computational-lambda}) with at least the following product
and sum type constructors:
\[ A,A_1,A_2,B ::= \tunit \s 0\s A_1 \ast A_2 \s A_1+A_2\s \dots 
\]
and terms, including the typical constructors and destructors but also
explicit sequencing ($\mathsf{let\,in}$)
\[\begin{array}{l@{}l}t,t_1,t_2,u ::= x &\s () \s (t_1,t_2) \s \pi_1\,t \s \pi_2\,t\s
  \inj_1\,t \s \inj_2\,t \\&\s \letin
  x {t_1} {t_2} \s\casezero t \s \case t {x_1}{u_1}{x_2}{u_2}\s \dots\end{array}
\]
We consider the standard typing rules (where $i\in\{1,2\}$):
  \begin{align*}
  &
   \infer{\Gamma, x : A, \Gamma' \vdash x : A}{} 
\ \quad   \infer{\Gamma \vdash () : \tunit}{}
  \ \quad \infer{\Gamma \vdash (t_1,t_2) : A_1 \ast A_2}{\Gamma \vdash t_1 : A_1 \quad \Gamma \vdash t_2 : A_2} 
   \ \quad \infer{\Gamma \vdash \pi_i\,t : A_i}{\Gamma \vdash t : A_1
    \ast A_2}
     \ \quad \infer{\Gamma \vdash \inj_i\,t : A_1+A_2}{\Gamma \vdash t : A_i}
   \\[6pt]
   &
     \infer{\Gamma \vdash \letin x t u : B}
     {\Gamma \vdash t : A \quad \Gamma, x : A \vdash u : B }
     \ \quad 
   \infer{\Gamma \vdash \casezero t : B}
           {\Gamma \vdash t : 0 }
     \ \quad 
   \infer{\Gamma \vdash \case t {x_1}{u_1}{x_2}{u_2}: B}
           {\Gamma \vdash t : A_1 + A_2\quad \big(\Gamma,x_i:A_i\vdash u_i:B\big)_{i\in\{1,2\}} }
  \end{align*}
  (Here, a context $\Gamma$ is a sequence of assignments of types~$A$ to variables~$x$.)

  In what follows, we use shorthands such as $\tbool=\tunit+\tunit$, and
  if-then-else instead of case. 

  This language is intended to be a generic probabilistic programming
  language, but so far there is nothing specifically probabilistic about this syntax.
  Different probabilistic programming languages support distributions
  over different kinds of structures. 
  Thus, our language is extended according to an `interface' by specifying
  type constants and typed term constants $f:A\to B$. For each term constant $f:A\to B$,
  we include a new typing rule,
 \[    \infer{\Gamma \vdash f(t) : B}
   {\Gamma \vdash t : A }\]
         \begin{example}\label{ex:interfaces}
           We consider the following examples of interfaces. 
  \begin{enumerate}
  \item For probabilistic programming over finite domains, we may have term constants such as
    $\tbernoulli_{0.5}:\tunit \to \tbool$, intuitively a fair coin toss.
  \item For probabilistic programming over real numbers, we may have a type constant $\treal$ and
    term constants such as $\tnormal:\treal\ast\treal \to\treal$,
    intuitively a parameterized normal distribution,
    and arithmetic operations such as
    $(+):\treal\ast\treal\to\treal$.
  \item The main interface of this paper is for random graphs: this
    has a type constant $\tvertex$ and term constants $
    \tnew:\tunit\to \tvertex$ and $\tedge:\tvertex\ast \tvertex \to \tbool$. 
\end{enumerate}\end{example}
 (We have kept this language as simple as possible, to focus on the interesting aspects.
 A practical probabilistic programming language will include other features, which are largely orthogonal,
 and indeed within our implementation in Haskell (\S\ref{sec:intro:practice}), programming features like higher order functions and recursion are
 present and useful. See also the discussion in~\S\ref{sec:dmc-to-am}.)

\subsection{Equational Theories and Markov Categories}
\label{sec:markov-cats}
Section~\ref{sec:generic-ppl} introduced a syntax for various probabilistic programming interfaces. 
The idea is that this is a generic language which applies to different
interfaces with different distributions that are implemented in
different ways. Rather than considering various ad hoc operational
semantics, we study the instances of interfaces by the program equations
that they support.

Regardless of the specifics of a particular implementation, we
expect basic equational reasoning principles for probabilistic
programming to hold, such as the
following laws:
\begin{align}
  &(\letinsqueeze y {(\letinsqueeze x t u)} {t'}) \equiv (\letinsqueeze x t {\letinsqueeze y u {t'}})\label{eq:let-assoc}
  && 
  (\text{where $x ∉ \fv(t')$})
\\
&  (t, u) ≡ (\letin {x}{t}{}\letin {y}{u}{(x,y)})
    \label{eq:moggi-let-pair}
  \\
  &(\letinsqueeze x t {\letinsqueeze {x'} {t'} u}) \equiv (\letinsqueeze {x'} {t'} {\letinsqueeze x t u})\label{eq:let-comm}&&(\text{where $x ∉ \fv(t')$ and $x' ∉ \fv(t)$})
\\\label{eq:let-affine}
   & (\letin x {t'} t) \equiv t && (\text{where $x ∉ \fv(t)$})
\intertext{The following law does not always hold, but does hold when
       $v$ is `deterministic'.}
         &(\letin x v t) \equiv t[v/x] \label{cd:let_val}
\end{align}
Equations~\eqref{eq:let-comm} and \eqref{eq:let-affine} say that parts of programs can be
re-ordered and discarded, as long as the dataflow is respected. This is a feature of
probabilistic programming. For example, coins do not remember
the order nor how many times they have been tossed. But these
equations would typically not hold in a language with state.

The cleanest way to study equational theories of programs is via a categorical semantics, and for Markov categories have arisen as a canonical setting for categorical probability. Informally, a category is a structure for composition, and this matches the composition structure of $\mathsf{let}$ in our language. We also have monoidal structure which allows for the type constructor $A\times B$ and for the compound contexts $\Gamma$, comonoid structure which allows duplication of variables, and distributive coproduct structure which allows for the sum types.

\begin{definition}\label{def:distr-markov}
  A \emph{symmetric monoidal category} $(\CatA,\otimes,I)$ is
  a category $\CatA$ equipped with a functor ${\otimes :\CatA\times \CatA \to\CatA}$ and an object
  $I$ together with associativity, unit and symmetry structure (\cite[XI.1.]{maclane}).
  A \emph{Markov category (\cite{fritz})} is a symmetric monoidal category in which
  \begin{itemize}
  \item the monoidal unit $I$ is a terminal object ($I=1$), and
  \item 
    every object $X$ is equipped with a comonoid $\Delta_X:X\to X\otimes X$,
    compatible with the tensor product ($\Delta_{X\otimes Y}=(X\otimes\mathsf{swp}\otimes Y)\cdot (\Delta_X\otimes \Delta_Y)$, where $\mathsf{swp}$ is the swap map of $\CatA$). 
  \end{itemize}
  A morphism $f\colon X\to Y$ in a Markov category is
  \emph{deterministic} if it commutes with the comonoids:
  $(f\otimes f)\cdot \Delta_X=\Delta_Y\cdot f$. 

  A \emph{distributive symmetric monoidal category (e.g.~\cite{walters,jay-distr-monoidal})} is a symmetric monoidal category equipped with chosen finite coproducts such that
  the canonical maps $X\otimes Z + Y\otimes Z\to (X+Y)\otimes Z$
    and $0\to 0\otimes Z$ are isomorphisms.
  A \emph{distributive Markov category} is a Markov category whose
  underlying monoidal category is also distributive and
  whose chosen coproduct injections $X\to X+Y\leftarrow Y$ are
  deterministic.
  A \emph{distributive category~\cite{clw,cockett-distr}} is a distributive Markov category
  where all morphisms are deterministic.
  
  A \emph{(strict) distributive Markov functor} is a functor $F : \CatA \to \CatB$ between distributive Markov categories which strictly preserves the chosen symmetric monoidal, coproduct, and comonoid structures.
\end{definition}
In this paper we mainly focus on functors between distributive Markov
categories that strictly preserve the relevant structure, so we elide
`strict'. (Nonetheless, non-strict functors are important,
e.g.~\cite[\S10.2]{fritz} and Prop.~\ref{prop:markov-embed-kleisli}.)

We interpret the language of Section~\ref{sec:generic-ppl} in a
distributive Markov category~$\CatA$ by interpreting types $A$ and
type contexts $\Gamma$ as objects $\sem A$ and $\sem \Gamma$, and
typed terms $\Gamma\vdash t :A $ as morphisms $\sem \Gamma\to \sem A$.
(See e.g.~\cite{pitts-cat-logic} for a general discussion of terms
as morphisms.) 

In more detail, to give such an interpretation, type constants must
first be given chosen interpretations as objects of $\CatA$.
We can then interpret types and contexts using the monoidal and coproduct structure of $\CatA$.
Following this, term constants $f:A\to B$ must be given chosen
interpretations as morphisms $\sem f:\sem A\to \sem B$ in $\CatA$.
The interpretation of other terms is made by induction on the
structure of typing derivations in a standard manner, using the
structure of the distributive Markov category (e.g.~\cite{bpdh},
\cite[\S7.2]{dario-thesis}).
For example,
\begin{align*}
  &\sem{\Gamma,x:A,\Gamma'\vdash x:A}
  =
  \sem{\Gamma,x:A,\Gamma'}
  \cong
  \sem{\Gamma}\otimes \sem A \otimes\sem{\Gamma'}
  \xrightarrow{!\otimes \sem A\otimes!}
  1\otimes
  \sem A\otimes 1
  \cong \sem A
\displaybreak[0]\\
&  \sem{\Gamma\vdash\letin x t u:B}
  =
  \sem\Gamma\xrightarrow{\Delta_{\sem\Gamma}} \sem\Gamma\otimes \sem\Gamma
  \xrightarrow{\sem\Gamma\otimes \sem t}
  \sem\Gamma\otimes \sem A
  =
  \sem{\Gamma,x:A}
  \xrightarrow{\sem u}
     \sem B
\displaybreak[0]  \\
  &
    \sem{\Gamma\vdash\case t {x_1}{u_1}{x_2}{u_2}:B}
    =\\&\qquad
    \sem \Gamma
    \xrightarrow{\Delta_{\sem\Gamma}}
    \sem\Gamma\otimes \sem\Gamma
    \xrightarrow{\sem \Gamma\otimes \sem t}
    \sem \Gamma\otimes {\sem {A_1+A_2}}
    \cong
    \sem{\Gamma ,x:A_1}+\sem{\Gamma,x:A_2}
    \xrightarrow{\langle\sem {u_1},\sem{u_2}\rangle}
  \sem B
  \\&\sem{\Gamma \vdash f(t):B}=
  \sem\Gamma\xrightarrow{\sem t}\sem A\xrightarrow{\sem f} \sem B
\end{align*}

An interpretation in a Markov category induces an equational theory
between programs: let $\Gamma\vdash t=u :A$ if $\sem t=\sem u$.
\begin{proposition}[e.g.~\cite{dario-thesis}, \S7.1]
The equational theory induced by the interpretation in a distributive Markov
category, with given interpretations of type and term constants,
always includes the
equations~\eqref{eq:let-assoc}--~\eqref{eq:let-affine}, 
and also~\eqref{cd:let_val} whenever $\sem v$ is a deterministic morphism. 
\end{proposition}
\begin{example}\label{ex:finset}
  The category $(\FinSet, ×, 1)$ of finite sets is a distributive Markov
    category. As in any category with products, each object has a
    unique comonoid structure, and all morphisms are deterministic. 
    This is a good Markov category for interpreting the plain language
    with no type or term constants. For example, $\sem{\tbool}$ is a
    set with two elements. 
  \end{example}\begin{example}\label{ex:finstoch} The category $\FinStoch$ has natural numbers as objects and
    the morphisms are stochastic matrices. In more detail, a morphism $m\to
    n$ is a matrix in $(\RR_{\geq 0})^{m\times n}$ such that each row sums
    to~$1$.
    Composition is by matrix multiplication. The monoidal structure is
    given on objects by multiplication of numbers, and on morphisms by
    Kronecker product of matrices.
    By choosing an enumeration of each finite set, we get a functor $\FinSet\to\FinStoch$ that converts a function to the corresponding
    $(0/1)$-valued matrix. So every object of $\FinStoch$ can be
    regarded with the comonoid structure from $\FinSet$. 
    The deterministic morphisms in
    $\FinStoch$ are exactly the morphisms from $\FinSet$~\cite[10.3]{fritz}. 

    This is a good Markov category for interpreting the language with
    Bernoulli distributions (Ex.~\ref{ex:interfaces}(1)). We
    interpret the fair coin as the $1\times 2$ matrix $(0.5,0.5)$.

    We can also give some interpretations for the graph interface (Ex.~\ref{ex:interfaces}(3)) in $\FinStoch$.
    For instance, consider random graphs made of two disjoint complete subgraphs,
    as is typical in a clustering model.
    We can interpret this by putting $\sem\tvertex =2$,
    $\sem\tedge=(\begin{smallmatrix}1&0&0&1
      \\0&1&1&0\end{smallmatrix})^\top$, and $\sem\tnew=(0.5,0.5)$.
\end{example}
We look at other examples of distributive Markov categories and
interpretations of these interfaces in
Sections~\ref{sec:distr-monad} and~\ref{sec:probthy}, and then in Sections~\ref{sec:graphons-to-equational-theories}--\ref{sec:ER-Rado}. 

\subsection{Equational Theories and Affine Monads}
\label{sec:affine-monad}
\subsubsection{Distributive Markov Categories from Affine Monads}
One way to generate equational theories via Markov categories is by
considering certain kinds of monads, following
Moggi~\cite{moggi-computational-lambda}.

\begin{definition}\label{def:monad}A \emph{strong monad} on a category $\CatC$
  with finite products is given by
  \begin{itemize}
  \item for each object $X$, an object $T(X)$;
  \item for each object $X$, a morphism $\eta_X:X\to T(X)$;
  \item for objects $Z,X,Y$, a family of functions natural in $Z$
    \[
      (\bind):\CatC(Z,T(X))\times \CatC(Z\times X,T(Y))\to \CatC(Z,T(Y))
    \]
  \end{itemize}
  such that $\bind$ is associative with unit $\eta$. 
\end{definition}
(There are various different formulations of this structure. 
When $\CatC$ is cartesian
closed, as in Defs.~\ref{def:distribution-monad} and~\ref{def:monad-radonom}, then the bind ($\bind$) is represented by a morphism 
$(\bind ) : {T(X)\times (X\Rightarrow T(Y))\to T(Y)}$, by the Yoneda lemma.)
\begin{definition}[\hspace{1sp}\cite{kock-comm,jacobs-weakening,lindner-affine}]\label{def:affine-monad}
  Given a strong monad~$T$, we say that two morphisms $f:X_1\to T(X_2)$,
  $g:X_1\to T(X_3)$
  \emph{commute} if
\[  f \bind ((g \circ \pi_1)\bind (\eta\circ \langle \pi_2 \circ \pi_1,\pi_2\rangle))
  {} =
    g \bind ((f \circ \pi_1)\bind (\eta\circ \langle \pi_2, \pi_2 \circ \pi_1\rangle))：X_1 → T(X_2 × X_3) \text .
\]
  A strong monad is \emph{commutative} if all morphisms commute. It is \emph{affine} if $T(1)\to 1$ is an isomorphism.
\end{definition}

The Kleisli category $\Kl(T)$ of a strong monad $T$ has the same objects as $\CatC$,
but the morphisms are different: $\Kl(T)(A,B)=\CatC(A,T(B))$.
There is a functor $J:\CatC\to\Kl(T)$, given on morphisms by composing
with $\eta$ (e.g.~\cite[\S VI.5]{maclane}, \cite{moggi-computational-lambda}).
\begin{proposition}\label{prop:distr-affine-monad}
        Let\, $T$ be a strong monad on a category $\CatC$. If\, $T$ is commutative
        and affine and $\CatC$ has finite products, then
the Kleisli category $\Kl(T)$ has a canonical structure of a Markov
category. Furthermore, if $\CatC$ is distributive, then $\Kl(T)$ can be regarded as a distributive Markov category. 
\end{proposition}
\begin{proof}[Proof notes]
  The Markov structure follows~\cite[\S3]{fritz}. Since $T$ is commutative, the product structure of $\CatC$ extends
  to a symmetric monoidal structure on $\Kl(T)$.
  Since $T(1)=1$, the monoidal unit ($1$) is terminal in $\Kl(T)$.
  Every object in $\CatC$ has a comonoid structure, and this is
  extended to $\Kl(T)$ via $J$. 
  The morphisms in the image of $J$ are deterministic, although this
  need not be a full characterization of determinism.
  
  For the distributive structure, recall that $J$ preserves
  coproducts and indeed it has a right adjoint. Hence, the coproduct injections will be deterministic.
\end{proof}
We can thus interpret the language of Section~\ref{sec:generic-ppl} using any
strong monad, interpreting the types~$A$ as objects~$\sem A$ of $\CatC$,
and a term $\Gamma\vdash t:A$ as a morphism
$\sem t:\sem{\Gamma} \to T(\sem A)$.
This interpretation matches Moggi's interpretation of the language of Section~\ref{sec:generic-ppl} in a strong monad. 

\subsubsection{Example Affine Monad: Distribution Monad}
\label{sec:distr-monad}
\begin{definition}[e.g.~\cite{jacobs-coalgebra}, \S4.1]\label{def:distribution-monad} The distribution monad $\DistM$ on $\Set$ is defined as follows:
\begin{itemize}
 \item On objects: each set $X$ is mapped to the set of all finitely-supported discrete probability measures on $X$, that is, all functions $p : X \to \RR$ that are non-zero for only finitely many elements and satisfy $\sum_{x \in X} p(x) = 1$.
 \item The unit $\eta_X :  X\to \DistM (X)$ maps $x \in X$ to the indicator function 
 $\lambda y.\, [y=x]$, i.e.~the Dirac distribution~$\delta_x$.
 \item The bind function $(\bind)$ is defined as follows: 
 \[
         (f \bind g)(z)(y) = \textstyle\sum_{x \in X} f(z)(x) \cdot g(z,x)(y)
 \] 
 \end{itemize}
\end{definition}
By the standard construction for strong monads, each morphism $f : X \to Y$ 
gets mapped to $\DistM f：\DistM X \to \DistM Y$, that is, the pushforward in this case: 
$\DistM f(p)(y) = \textstyle\sum_{x \in f^{-1}(y)} p(x)$.
Consider the language with no type constants, and just the term constant $\tbernoulli_{0.5}$ (Ex.~\ref{ex:interfaces}(1)).
This can be interpreted in the distribution monad.
Since every type $A$ is interpreted as a finite set~$\sem A$, and every context $\Gamma$ as a finite set~$\sem\Gamma$,
a term $\Gamma \vdash t:A$ is interpreted as a function $\sem \Gamma\to \DistM\sem A$.
To give a Kleisli morphism between finite sets is to give a stochastic matrix, and so the induced equational theory is the same as
the interpretation in $\FinStoch$ (Ex.~\ref{ex:finstoch}). 
\subsubsection{Example Affine Monad: Giry Monad}
\label{sec:probthy}
We recall some rudiments of measure-theoretic probability.
\begin{definition}
  A \emph{$\sigma$-algebra} on a set is a non-empty collection of subsets that contains the empty set and is closed under countable unions and complements. A \emph{measurable space} is a pair $(X,\Sigma)$ of a set and a $\sigma$-algebra on it.
  A measurable function $(X,\Sigma_X)\to (Y,\Sigma_Y)$ is a function $f\colon X\to Y$
  such that $f\inv(U)\in \Sigma_X$ for all $U\in \Sigma_Y$.
  
  A \emph{probability measure}
  on a measurable space $(X,\Sigma)$ is a function $\mu:\Sigma\to[0,1]$ that has total mass $1$ ($\mu(X)=1$) and that is $\sigma$-additive:
  $\mu(\biguplus_{i=1}^\infty U_i)=\sum_{i=1}^\infty \mu(U_i)$ for any sequence of disjoint $U_i$.
\end{definition}
Examples of measurable spaces include: the finite sets $X$ equipped with their powerset $\sigma$-algebras; the unit interval $[0,1]$ equipped with its Borel $\sigma$-algebra, which is the least $\sigma$-algebra containing the open sets.
Examples of probability measures include: discrete probability measures (Def.~\ref{def:distribution-monad}); the uniform measure on $[0,1]$; the Dirac distribution $\delta_x(U)={[x\in U]}$.

The product of two measurable spaces $(X,\Sigma_X)\times (Y,\Sigma_Y)=(X\times Y,\Sigma_X\otimes \Sigma_Y)$ comprises the product of sets with the least $\sigma$-algebra making the projections $X\leftarrow X\times Y\to Y$ measurable.
The category of measurable spaces and measurable functions is a distributive category. 

A \emph{probability kernel} between measurable spaces $(X,\Sigma_X)$ and $(Y,\Sigma_Y)$ is a function $k\colon X\times \Sigma_Y\to [0,1]$ that is measurable in the first argument and that is $\sigma$-additive and has mass $1$ in the second argument. 

To compose probability kernels, we briefly recall Lebesgue integration. Consider a measurable space $(X,\Sigma_X)$, 
a measure $\mu:\Sigma_X\to[0,1]$, and a measurable function $f\colon X\to [0,1]$. If $f$ is a simple function, i.e. $f(x) = \sum_{i = 1}^m r_i \cdot [x \in U_i]$ for some $m$, $r_i \in [0,1]$, and $U_i \in \Sigma_X$, the Lebesgue integral $\int f\,\dd \mu = \int f(x)\,\mu(\dd x)\in[0,1]$ is defined to be $\sum_{i = 1}^m r_i \times \mu(U_i)$. If $f$ is not a simple function, there exists a sequence of increasing simple functions $f_1,f_2,\ldots : X \to [0,1]$ such that $\sup_k f_k(x) = f(x)$ (for example, by taking $f_k(x) ≝ \lfloor 10^k f(x) \rfloor / 10^k$). In that case, the integral is defined to be the limit of the integrals of the $f_k$'s (which exists by monotone convergence).


Probability kernels can be equivalently formulated as morphisms $X\to \Giry (Y)$, where $\Giry$ is the Giry monad:
\begin{definition}[\hspace{1sp}\cite{giry}]\label{def:giry}
  The Giry monad $\Giry$ is a strong monad on the category $\Meas$ of measurable spaces given by
  \begin{itemize}
\item  $\Giry(X)$ is the set of probability measures on $X$, with the least $\sigma$-algebra making $\int f\, \dd(-):\Giry(X)\to [0,1]$ measurable for all measurable $f:X\to [0,1]$;
\item the unit $\eta$ maps $x$ to the Dirac distribution $\delta_x$;
\item the bind is given by composing kernels:
  \begin{equation}
  (k\bind l) (z,U) = \int l((z,x),U) \, k(z,\dd x)\text.
  \end{equation}
  \end{itemize}
\end{definition}
\begin{proposition}\label{prop:giry-comm}
  The monad $\Giry$ is commutative and affine.
\end{proposition}
\begin{proof}[Proof notes] Commutativity boils down to Fubini's theorem for reordering integrals and affineness is marginalization (since probability measures have mass $1$). See also~\cite{jacobs-commutative-effectus}.\end{proof}

\label{giry-real-lang}
Consider the real-numbers language (Ex.~\ref{ex:interfaces}(2)).
Let $\sem \treal=\RR$, with the Borel sets, and interpret $\tnormal$ as the normal probability measure on $\RR$. The basic arithmetic operations are all measurable.

Among the following three programs
  \begin{align}
    &\letin x {\tnormal(0,1)} {x+x} \label{eqn:norm-eg-1}\\
    &\letin x {\tnormal(0,1)} {\letin y {\tnormal(0,1)} {x+y} }\label{eqn:norm-eg-2}\\
    &\tnormal(0,1)+\tnormal(0,1) \label{eqn:norm-eg-3}\end{align}
  the programs $\eqref{eqn:norm-eg-2}$ and $\eqref{eqn:norm-eg-3}$
  denote the same normal distribution with variance~$2$, whereas
  \eqref{eqn:norm-eg-1} denotes a distribution with variance~$4$.
  Notice that we cannot use~\eqref{cd:let_val} to equate all the
  programs, because $\sem{\tnormal}$ is not deterministic.

We can also interpret the Bernoulli language (Ex.~\ref{ex:interfaces}(1)) in the Giry monad; this interpretation gives the same equational theory as the interpretation in $\FinStoch$ and in the distribution monad in Section~\ref{sec:distr-monad}.

    We can also give some interpretations for the graph interface
    (Ex.~\ref{ex:interfaces}(3)) in the Giry monad.
    For an informal example, consider the geometric example from Section~\ref{sec:intro:graphs}, let $\sem{\tvertex}=S_2$ (the sphere), and define $\sem{\tnew}$ to be the uniform distribution on the sphere.
    (See also Section~\ref{sec:all-bw}.) 

    \subsubsection{Affine Monads from Distributive Markov Categories}\label{sec:dmc-to-am}
    The following result, a converse to Proposition~\ref{prop:distr-affine-monad}, demonstrates that the new notion of distributive Markov category (Def.~\ref{def:distr-markov}) is a canonical one, and emphasizes the close relationship between semantics with distributive Markov categories and semantics with commutative affine monads. 
\begin{proposition}\label{prop:markov-embed-kleisli}
  Let $\CatA$ be a small distributive Markov category.
  Then, there is a distributive category~$\CatC$ with a commutative affine monad $T$ on it
  and a full and faithful functor $\CatA\to \Kl(T)$ that preserves
  symmetric monoidal structure, comonoids, and sums.
\end{proposition}
\begin{proof}[Proof notes]
  Our proof is essentially a recasting of \cite[\S 7]{power-universal} to this different
  situation, as follows. 
  
  Let $\CatADet$ be the wide subcategory of $\CatA$ comprising the
  deterministic morphisms, and write $J:\CatADet\to\CatA$ for the
  identity-on-objects inclusion functor. Note that $\CatADet$ is a distributive category. We would
  like to exhibit $\CatA$ as the Kleisli category for a monad on
  $\CatADet$, but this might not be possible: intuitively, $\CatADet$
  might be too small for the monad to exist. Instead, we first embed
  $\CatADet$ in a larger category~$\CatC$ and construct a monad on $\CatC$.

  The main construction in our proof is the idea that if $\CatX$ is a
  small distributive monoidal category, then the category
  $\FP(\CatX\op,\Set)$
  of finite-product-preserving functors is such that
  \begin{itemize}
  \item   $\FP(\CatX\op,\Set)$ is cocomplete and moreover total (\cite{STREET1978350}) as a
    category;
  \item $\FP(\CatX\op,\Set)$ admits a distributive monoidal structure;
  \item   the Yoneda embedding $\CatX\to[\CatX\op,\Set]$, which is
    full and faithful, factors through
    $\FP(\CatX\op,\Set)$,
    and this embedding $\CatX\to \FP(\CatX\op,\Set)$ preserves finite sums and is
    strongly monoidal;
  \item  the Yoneda embedding exhibits $\FP(\CatX\op,\Set)$ as a free
    colimit completion of $\CatX$ as a monoidal category that already has
    finite coproducts.
  \end{itemize}

  So we let $\CatC=\FP(\CatADet\op,\Set)$ comprise the
  finite-product-preserving functors $\CatADet\op\to \Set$.
  This is a distributive category. 
  To get a monad on $\CatC$, we note that since
  $\FP(\CatA\op,\Set)$ has finite coproducts and
  $\CatADet\to \CatA\to \FP(\CatA\op,\Set)$  preserves finite
  coproducts and is monoidal, the
  monoidal structure
  induces a canonical colimit-preserving monoidal functor
  $J_!:\FP(\CatADet\op,\Set)\to \FP(\CatA\op,\Set)$.
  Any colimit-preserving functor~$J_!$ out of a total category 
  has a right adjoint~$J^*$, and hence a monoidal monad~$(J^*J_!)$ is induced on
  $\CatC$. 
  
  It remains for us to check that the embedding $\CatA\to
  \FP(\CatA\op,\Set)$ factors through the comparison functor
  $\Kl(J^*J_!)\to \FP(\CatA\op,\Set)$, which follows from the fact
  that $J:\CatADet\to\CatA$ is identity on objects.
\end{proof}
As an aside, we note that, although our simple language in Section~\ref{sec:generic-ppl} did not include higher-order functions, the category $\CatC$ constructed in the proof of Proposition~\ref{prop:markov-embed-kleisli} is cartesian closed, and since the embedding is full and faithful, this shows that higher-order functions would be a conservative extension of our language. Indeed, this kind of conservativity result was part of the motivation of~\cite{power-universal}.
For the same reason, inductive types (lists, and so on) would also be a conservative extension. 
We leave conservativity with other language features for future work. Recursion in probabilistic programming is still under investigation~\cite{DBLP:conf/lics/JiaLMZ21,DBLP:conf/lics/MatacheMS22,DBLP:journals/pacmpl/VakarKS19,DBLP:journals/corr/abs-2106-16190,DBLP:journals/pacmpl/EhrhardPT18}; there is also the question of conservativity with respect to combining Markov categories, e.g. combining real number distributions (\eqref{eqn:tnormal}--\eqref{eqn:add}) with graph programming (\eqref{eqn:intro-new}--\eqref{eqn:intro-edge}).

\subsection{Bernoulli Bases, Numerals and Observation}\label{sec:bernoulli-base}
Although an interface may have different type constants, it will
always have the `numeral' types, sometimes called `finite' types:
\[
0\quad  \tunit \quad \tbool=\tunit+\tunit \quad \tunit+\tunit+\tunit\quad\dots
\]
For probabilistic programming languages, there is a clear expectation
of what will happen when we run a program of type $\tbool$: it will
randomly produce either $\ttrue$ or $\tfalse$, each with some
probability. Similarly for other
numeral types. For type constants, we might not have evident notions of
observation or expected outcomes. But for numeral types, it should be
routine. We now make this precise via the notion of Bernoulli base. 

On the semantic side, distributive Markov categories will always have `numeral' objects
\[
  0\quad 1\quad 2\defeq1+1\quad 3\defeq 1+1+1\quad\dots
\]
For any type~$A$ formed without type constants, and any Markov
category, we have
that $\sem A\cong n$ for some numeral object.
Any equational theory for the programming language induces in
particular an equational theory for the sub-language without any type
constants.

\begin{proposition}
  For any distributive Markov category $\CatA$, let $\CatA_\NN$ be the category whose objects are natural numbers, and where the morphisms are the morphisms in $\CatA$ between the corresponding numeral objects.
  This
  is again a distributive Markov category.
\end{proposition}

\begin{example}\label{ex:bernoulli-base}\begin{enumerate}
  \item $\FinSet_\NN = \Set_\NN$ is equivalent to $\FinSet$ as a category.
  \item For the finite distributions and the Giry monad (\S\ref{sec:distr-monad}--\ref{sec:probthy}), $\Kl(\DistM)_\NN\simeq \Kl(\Giry)_\NN\simeq \FinStoch$.\end{enumerate}
\end{example}

Recall that a functor is \emph{faithful} if it is injective on hom-sets.
\newcommand{\faithfulfr}{\rightarrowtail}
\begin{definition}
  A \emph{Bernoulli base} for a distributive Markov category $\CatA$
  is a faithful distributive Markov functor $\bbase:\CatA_\NN\faithfulfr \FinStoch$.
\end{definition}
Thus, for any distributive Markov category with a Bernoulli base, 
for any closed term $\vdash t : A$ of numeral type ($\sem A=n$), we
can regard its interpretation $\sem t:1\to n$ as nothing but 
a probability distribution $\bbase(\sem t)$ on $n$ outcomes. This is the case even if
$t$ uses term constants and has intermediate subterms using type
constants. 

\begin{example}
  All the examples seen so far can be given Bernoulli bases. In fact, for $\FinStoch$, $\Kl(\DistM)$ and $\Kl(\Giry)$,
  the functor $\bbase:\CatA_\NN\faithfulfr \FinStoch$ is an isomorphism of distributive Markov categories.
\end{example}
When $\bbase$ is an isomorphism of categories, that means that \emph{all} the finite probabilities are present in~$\CatA$. This is slightly stronger than we need in general. 
For instance, when $\CatA=\FinSet$, there is a unique Bernoulli base $\bbase:\FinSet_\NN\faithfulfr\FinStoch$, taking a function to a $0/1$-valued matrix,
but it is not full.
We could also consider variations on $\FinStoch$. For example,
consider the subcategory $\mathbf{Fin\mathbb{Q}Stoch}$ of $\FinStoch$ where the matrices are rational-valued;
this has a Bernoulli base that is not an isomorphism.

\subsection{Quotients of Distributive Markov Categories}\label{sec:quotients}
\newcommand{\Ctx}{\mathcal{C}}

We provide a new, general method for constructing a Bernoulli-based Markov
category out of a distributive Markov category.
Our construction is a categorical formulation of the notion of
contextual equivalence.

Recall that, in general, contextual equivalence for a programming language starts
with a notion of basic observation for closed programs at ground
types. We then say that programs $\Gamma\vdash t,u:A$ at other types
are \emph{contextually equivalent} if for every context $\Ctx$ with $\vdash\Ctx[t],\Ctx[u]:n$, for some ground type $n$,
we have that $\Ctx[t]$ and $\Ctx[u]$ satisfy the same observations. 
In the categorical setting, the notion of observation is given by a
distributive Markov functor $\CatA_\NN\to\FinStoch$, and the notion of
context~$\Ctx$ is replaced by suitable morphisms ($h$, $k$ below). 
We now introduce a quotient construction that will be key in showing that every graphon arises from a distributive Markov category (Corollary~\ref{corollary:graphon-markov}), via Theorem~\ref{thm:eq-theory-to-graphon}.
We note that this is a general new method for building Markov categories.

\begin{proposition}\label{prop:quotient}
Let $\CatA$ be a distributive Markov category, and let $\bbase\colon
\CatA_\NN\to \FinStoch$ be a distributive Markov functor.
Suppose that for every object $X\in\CatA$, either $X=0$ or there
exists a morphism $1\to X$. 
Then, there is a distributive Markov category $\CatAp$ with a Bernoulli base, equipped with 
a distributive Markov functor $\CatA\to \CatAp$ and a factorization of distributive Markov functors
$\bbase=\CatA_\NN\to (\CatAp)_\NN\faithfulfr \FinStoch$.
\end{proposition}
\begin{proof}
  Define an equivalence relation~$\sim$ on each hom-set $\CatA(X,Y)$,
  by $f\sim g:X\to Y$ if 
  \[
    \forall Z,n.\,\forall h:1\to X\otimes Z.\,\forall k:Y\otimes Z\to
    n.\,
    \quad
    \bbase(k\cdot (f\otimes Z)\cdot h)=
    \bbase(k\cdot (g\otimes Z)\cdot h)
    \text{ in }\FinStoch(1,n)\text.
  \]
  Informally, our equivalence relation considers all ways of generating $X$'s
  via precomposition ($h$),  all ways for testing $Y$'s via postcomposition ($k$), and all ways of combining with some ancillary data ($Z$).
  It is essential that we consider all these kinds of composition in order for the quotient category to have the categorical structure. 
  
  It is immediate that composition of morphisms respects $\sim$,
  and hence we have a category: the objects are the same as $\CatA$,
  and the morphisms are $\sim$-equivalence classes. This is our category $\CatAp$.

  It is also immediate that if $f\sim g$ and $f'\sim
  g'$ then $(f\otimes f')\sim(g\otimes g')$. Thus, $\CatAp$ is a
  monoidal category. 

  For the coproduct structure, we must show that if $f\sim g:X\to Y$ and
  $f'\sim g':X'\to Y'$ then $(f+f')\sim (g+g'):X+X'\to Y+Y'$.
  We proceed by noting that since we have morphisms
  $x:1\to X$ and $x':1\to X'$, as well as terminal morphisms $X\to 1$ and
  $X'\to 1$, we have that $X+X'$ is a retract of $X\otimes X'\otimes
  2$, with the section and retraction given by:
  \[
    X + X'\xrightarrow{X\otimes x'+x\otimes X} X\otimes X' + X\otimes
    X'\cong X\otimes X'\otimes 2
    \quad
     X\otimes X'\otimes 2\cong
    X\otimes X' + X\otimes X'
    \xrightarrow{X\otimes !+!\otimes X}
    X + X'
  \]
  Thus, by composing with this retract, it suffices to check that
  $(f\otimes f'\otimes 2)\sim (g\otimes g'\otimes 2)$,
  which we have already shown.

  The functor $\bbase:\CatA_\NN\to \FinStoch$ clearly factors through $(\CatAp)_\NN$,
  but it remains to check that the functor $(\CatAp)_\NN\to \FinStoch$ is now
  faithful (Bernoulli base). 
  So suppose that $\bbase(f)=\bbase(g)$. To show that $f\sim g :1\to
  m$, 
  we consider $h:1\to 1\otimes Z$, and $k:m\otimes Z\to n$.
  We must show that $\bbase(k\cdot (f\otimes Z)\cdot h)=\bbase(k\cdot (g\otimes Z)\cdot
  h)$.
  Since $h=1\otimes h'$, for some $h':1\to Z$,
  we have
  \begin{align*}
    \bbase(k\cdot (f\otimes Z)\cdot h)&=\bbase(k\cdot(m\otimes h')\cdot f)
    =
    \bbase(k\cdot(m\otimes h'))\cdot \bbase(f)\\
    &=
    \bbase(k\cdot(m\otimes h'))\cdot \bbase(g)
    =
    \bbase(k\cdot(m\otimes h')\cdot g)
    =
    \bbase(k\cdot (g\otimes Z)\cdot h)\text.\end{align*}
\end{proof}


%% file: equations_to_graphons.tex
\section{From program equations to graphons}\label{sec:prog-eq-to-graphons}

The graph interface for the probabilistic programming language (Ex.~\ref{ex:interfaces}(3))
does not have one fixed equational theory. Rather, we want to consider different equational theories for the language, corresponding to different implementations of the interface for the graph (see also \S\ref{sec:intro:eqthys}).
We now show how the different equational theories for the graph language each
give rise to a graphon, by building adjacency matrices for finite graphs~(shown in \eqref{eqn:programgraph}). To do this, we set up the well-behaved equational theories (\S\ref{sec:bernoulli-base}), recall the connection between graphons and finite random graphs (\S\ref{sec:graphons}), and then show the main result (\S\ref{sec:program-equivalence-graphons}, Theorem~\ref{thm:eq-theory-to-graphon}).

\subsection{Graphons as Consistent and Local Random Graph Models}
\label{sec:graphons}
For all $n ≥ 1$, let $[n]$ be the set $\{1,…,n\}$. (We sometimes omit
the square brackets, when it is clear.)
A simple undirected graph $g$ with $n$ nodes can be represented by its adjacency matrix $A_g\in 2^{[n]^2}$ such that $A_g(i,i)=0$ and $A_g(i,j)= A_g(j,i)$. Henceforth, we will assume that finite graphs are simple and undirected, unless otherwise stated. 
A random finite graph, then, has a probability distribution in
$\DistM\big(2^{[n]^2}\big)$ that only assigns non-zero probability to adjacency matrices.

\begin{definition}[e.g.~{\cite[\S11.2.1]{MR3012035}}]\label{def:rgm}
  A \emph{random graph model} is a sequence of distributions of random finite graphs of the form:
\begin{equation*}
  p_1\in \DistM\big(2^{[1]^2}\big),\,
  p_2\in \DistM\big(2^{[2]^2}\big),\,
  \dots,\, 
  p_n\in \DistM\big(2^{[n]^2}\big),
  \dots 
\end{equation*}

We say such a sequence is
\begin{itemize}
\item \emph{exchangeable} if each of its elements is invariant under permuting nodes:
  for every $n$ and bijection $\sigma \colon [n]\to [n]$, we have 
  $\DistM\big(2^{(\sigma^2)}\big)(p_n)=p_n$ (where
  $2^{(\sigma^2)}: 2^{[n]^2}\to 2^{[n]^2}$ is the function that
  permutes the rows and columns according to $\sigma$;
  we are regarding $\DistM$ as a covariant functor, Def.~\ref{def:distribution-monad}, and $2^{(-)}$ as a contravariant functor);
\item \emph{consistent} if the sequence is related by marginals:
  for every $n$ and for the inclusion function $\iota：[n] ↪ [n+1]$,
  $\DistM\big(2^{(\iota^2)}\big)(p_{n+1})=p_n$ (where
  $2^{(\iota^2)}:2^{([n+1]^2)}\to 2^{[n]^2}$ is the evident projection);
\item \emph{local} if the subgraphs are independent:
  if $A \subseteq [n]$ and $B \subseteq [n]$ are disjoint,
  then we have an injective function ${\jmath\colon A^2 + B^2\hookrightarrow [n]^2}$, and
  $\DistM\big(2^{\jmath}\big)(p_{n})\in \DistM\big(2^{(A^2)}\times 2^{(B^2)}\big)$
  is a product measure $p_A\otimes p_B$
(where \raisebox{0pt}[0pt]{$2^{\jmath}:2^{[n]^2}\to 2^{(A^2)}\times 2^{(B^2)}$} is
the evident pairing of projections).

\end{itemize}
\end{definition}

 \begin{definition}[e.g.~\cite{MR3012035}]
        A \emph{graphon} $W$ is a symmetric measurable function $W : [0,1]^2 \to [0,1]$.
\end{definition}
\newcommand{\graphonrgm}[3]{p_{#1,#2}(A_{#3})}
Given a graphon $W$, we can generate a finite simple undirected graph
$g$ with vertex set $[n]$ by sampling $n$ points $x_1, \ldots, x_n$
uniformly from $[0,1]$ and, then, including the edge $(i, j)$ with
probability $W(x_i, x_j)$ for all $1 ≤ i, j ≤ n$. This sampling
procedure defines a distribution over finite graphs: the probability
$\graphonrgm W n g$ of the graph $g = ([n], E)$ is:
\begin{equation}\int_{[0,1]^n} 
\!\prod_{(i,j) ∈ E} \!\!\! W(x_i, x_j) 
\prod_{(i,j) ∉ E} \!\!\!\! \left(1 - W(x_i, x_j) \right)
\dd(x_1\ldots x_n)\label{eqn:pwg}\end{equation}


\begin{proposition}[\hspace{1sp}\cite{lovasz-2006}, {\cite[\S11.2]{MR3012035}}]
  \label{prop:lovasz}
  Every graphon generates an exchangeable, consistent, and local random
  graph model, by the sampling procedure of \eqref{eqn:pwg}.
  Conversely, every exchangeable, consistent, and local random graph
  model is of the form~$p_{W,n}$ for some graphon~$W$.
\end{proposition}
\begin{proof}[Note]
 There are various methods for constructing $W$ from an exchangeable, consistent and local random graph model, however all are highly non-trivial. A general idea is that $W$ is a kind of limit object. For examples see e.g.~\cite[\S11.3]{lovasz-2006} or \cite{tao-graphons}. Fortunately though, we will not need explicit constructions in this paper. 
\end{proof}

\subsection{Theories of Program Equivalence Induce Graphons}
\label{sec:program-equivalence-graphons}
In this section we consider the instance of the generic language with the graph interface
(Ex.~\ref{ex:interfaces}(3)):
\[
  \tvertex\qquad
  \tnew : \tunit \to \tvertex\qquad
  \tedge:\tvertex\ast\tvertex\to \tbool
\]
We consider a theory of program equivalence, i.e.~a distributive Markov
category with a distinguished object $\sem\tvertex$
and morphisms $\sem \tnew :1\to \sem \tvertex$ and $\sem \tedge:\sem \tvertex\otimes
\sem \tvertex\to 1+1$.
We make two assumptions about the theory:
\begin{itemize}
\item  The graphs are simple and undirected:
\begin{equation}
  \begin{aligned}
  x\colon \tvertex \vdash \tedge(x,x)\equiv \tfalse
\qquad x,y\colon \tvertex \vdash \tedge(x,y)\equiv \tedge(y,x)
\end{aligned}\label{eqn:simplegraph}\end{equation}
and $\tedge$ is deterministic.
\item The theory is Bernoulli based (\S\ref{sec:bernoulli-base}).\end{itemize}
For each $n\in\NN$, we can build a
random graph with $n$ vertices as follows.
We consider the following  program $t_n$:
\begin{equation}\begin{aligned}
\vdash\ & \letin {x_1}{\tnew()}{
        \dots {} }\letin {x_n}{\tnew()}{}
    \begin{pmatrix}\tedge (x_1,x_1)&{\dots}&\tedge (x_1,x_n)
\\
\vdots &&{\vdots}\\
\tedge(x_n,x_1)&\dots &\tedge(x_n,x_n)\end{pmatrix}:\tbool^{(n^2)}
\end{aligned}\label{eqn:programgraph}\end{equation}
(Here we use syntactic sugar, writing a matrix instead of iteratively using pairs.)

Because the equational theory is Bernoulli-based, the
interpretation $\sem {t_n}$ induces a probability distribution $\bbase\sem{t_n}$
on~$2^{(n^2)}$. For clarity, we elide $\bbase$ in what follows, since
it is faithful. 
\begin{proposition}
  Each random matrix in \eqref{eqn:programgraph} is a random adjacency matrix, \ie a random graph. 
\end{proposition}
\begin{proof}[Proof note]This follows from (\ref{eqn:simplegraph}).\end{proof}
\begin{theorem}\label{thm:eq-theory-to-graphon}
  For any Bernoulli-based equational theory, the random graph model $(\sem {t_n})_n$ in \eqref{eqn:programgraph}
  is exchangeable, consistent, and local. Thus, the equational theory
  induces a graphon. 
\end{theorem}
\begin{proof}[Proof] We denote the matrix in~\eqref{eqn:programgraph} by $(\tedge (x_i,x_j))_{i, j ∈ [n]}$.
\paragraph{Exchangeability} We show that the distribution $\sem {t_n}$ is invariant under relabeling the nodes. By commutativity of the $\mathsf{let}$ construct~\eqref{eq:let-comm}, the program
 \begin{align*}
 & t_n^σ ≝ 
 \letin {x_{σ^{-1}(1)}}{\tnew()}{\ldots\letin {x_{σ^{-1}(n)}}{\tnew()}{}}
(\tedge (x_i,x_j))_{i, j \in [n]}
 \end{align*}
 satisfies $\sem{t_n^σ} = \sem{t_n}$. Hence, $\DistM(2^{\sigma^2})(\sem{t_n}) = \sem{t_n^σ} = \sem{t_n}$, for every $n$ and bijection $σ：[n] → [n]$.

\paragraph{Consistency}
    We define a macro $\subm_I$ in the graph programming language to extract a submatrix at the index set $I ⊆ [n]$: we have the (definitional) equality
    $$\subm_I((a_{i,j})_{i, j ∈ [n]}) ≝ (a_{i,j})_{i, j ∈ I} \quad \text{for } I ⊆ [n].$$
    We need to show that, if we delete the last node from a graph sampled from $\sem{t_{n+1}}$, the resulting graph has distribution $\sem{t_n}$. 
    This amounts to the affineness property (\ref{eq:let-affine}), as follows. 
    Let $g \sim \sem{t_{n+1}}$ be a random graph, and let $g' ≝ g\restrict{[n]}$ be the graph obtained by deleting the last node from $g$. Then clearly, the adjacency matrix of $g'$ is the adjacency matrix of $g$ where the last row and column have been removed, \ie $g'$ is sampled from the interpretation of the program:
    \begin{align*}
    t' &≝ \letin {x_1}{\tnew()}{\dots {} }\letin {x_n}{\tnew()}{}
         \letin {x_{n+1}}{\tnew()}{} \subm_{[n]}\big((\tedge(x_i,x_j))_{i,j ∈ [n+1]}\big)\\
    &\equiv \letin {x_1}{\tnew()}{\dots {}}\letin {x_n}{\tnew()}{}
         \letin {x_{n+1}}{\tnew()}{} (\tedge(x_i,x_j))_{i,j ∈ [n]}\\
            &\equiv \letin {x_1}{\tnew()}{\dots {}}\letin {x_n}{\tnew()}{} 
    (\tedge(x_i,x_j))_{i,j∈[n]} \hspace{3em} \text{(by \eqref{eq:let-affine})}\\ 
    &\equiv t_n.
    \end{align*}

\paragraph{Locality}
  Without loss of generality (by exchangeability and consistency), we need to show that for every random graph $g \sim \sem{t_n}$ and $1 < k < n$, the subgraphs $g_{A_k}, \, g_{B_k}$ respectively induced by the sets $A_k ≝ [k]$ and $B_k ≝ \{k+1,… , n\}$ are independent as random variables. 
  Let $\jmath$ be the injection $\jmath：A_k^2+B_k^2 ↪ n^2$, and $g' \sim \DistM(2^{\jmath})(\sem{t_n}) ∈ \DistM(2^{(A_k^2)}×2^{(B_k^2)})$. We want to show that $g'$ and $(g_{A_k}, g_{B_k})\sim \sem{t_k} ⊗ \sem{t_{n-k}}$ (by consistency) are equal in distribution.
  Modulo $α$-renaming,
  $(g_{A_k}, g_{B_k})$ is sampled from the interpretation of the
  program:
  \begingroup
  \allowdisplaybreaks
  \begin{align*}
  t' &≝ \big(\letin {x_1}{\tnew()}{\dots{} }\letin {x_k}{\tnew()}{}(\tedge(x_i,x_j))_{i,j ∈ [k]}, \\*
  & \hspace*{2em} \letin {x_{k+1}}{\tnew()}{\dots{} }\letin {x_n}{\tnew()}{} (\tedge(x_i,x_j))_{k+1 ≤ i,j ≤ n}\big)\\
  & ≡ \letin {u_1}{(\letin {x_1}{\tnew()}{\dots{} }\letin {x_k}{\tnew()}{}
    (\tedge(x_i,x_j))_{i,j ∈ A_k})}{}\\*
  & \phantom{{} \equiv {}} \letin {u_2}{(\letin {x_{k+1}}{\tnew()}{\dots{}}\letin {x_n}{\tnew()}{}(\tedge(x_i,x_j))_{i,j ∈ B_k})}{}(u_1, u_2) &&\text{(by \eqref{eq:moggi-let-pair})}\\
  &≡ \letin {x_1}{\tnew()}{\dots{} }\letin {x_k}{\tnew()}{}\letin {x_{k+1}}{\tnew()}{\dots{}}\letin {x_n}{\tnew()}{}&&\text{(\eqref{eq:let-assoc},\eqref{eq:let-comm})}\\*
  & \hspace*{1.3em} \letin {u_1}{\subm_{A_k}\big((\tedge(x_i,x_j))_{i,j ∈ [n]}\big)}{} 
\letin{u_2}{\subm_{B_k}\big((\tedge(x_i,x_j))_{i,j ∈ [n]}\big)}{} 
    (u_1, u_2)\hspace{-4cm}
      \\
  &≡ \letin {x_1}{\tnew()}{\dots{}}\letin {x_n}{\tnew()}{}\\*
  & \hspace*{1.3em} \letin {t}{(\tedge(x_i,x_j))_{i,j ∈ [n]}}{} \letin {u_1}{\subm_{A_k}(t)}{}\letin {u_2}{\subm_{B_k}(t)}{}(u_1, u_2) &&\text{(by \eqref{cd:let_val})}\\
  &≡ \letin {x_1}{\tnew()}{\dots{}}\letin {x_n}{\tnew()}{}\\*
  & \hspace*{1.3em} \letin {t}{(\tedge(x_i,x_j))_{i,j ∈ [n]}}{} 
\big(\subm_{A_k}(t),\, \subm_{B_k}(t)\big) &&\text{(by \eqref{eq:moggi-let-pair})}
  \\
  &≡ \letin {t}{\big(\letin {x_1}{\tnew()}{\dots{}} 
    \letin {x_n}{\tnew()}{} (\tedge(x_i,x_j))_{i,j ∈ [n]}\big)}{}\\*
& \hspace*{1.6em}     \big(\subm_{A_k}(t), \, \subm_{B_k}(t)) &&\text{(by \eqref{eq:let-assoc})}
  \end{align*}
  \endgroup
  and $g' \sim \DistM(2^{\jmath})(\sem{t_n})$ is indeed sampled from the
  interpretation of the latter program, 
  which yields the result.
\end{proof}


%% file: graphons_to_equations_syntactic.tex
\section{From graphons to program equations}\label{sec:graphons-to-equational-theories}

\newcommand{\CatCLG}{\mathbf{RGM}}
\newcommand{\Inj}{\mathbf{FinSetInj}}\label{sec:rgm-cat}
\newcommand{\Nat}{\mathrm{Nat}}

In Section~\ref{sec:prog-eq-to-graphons}, we showed how a distributive Markov category modelling the graph interface (Ex.~\ref{ex:interfaces}(3)) gives rise to a graphon.
In this section, we establish a converse: every graphon arises in this
way (Corollary~\ref{corollary:graphon-markov}). 
Theorem~\ref{thm:Markov-functor-graphon} will establish slightly more:
there is a `generic' distributive Markov category
(\S\ref{sec:generic-dmc}) modelling the graph interface whose
Bernoulli-based quotients are in precise correspondence with graphons
(\S\ref{sec:bernoulli-base-graphon}). This approach also suggests an
operational way of implementing the graph interface for any graphon (\S\ref{sec:abstract-op-sem}).

\subsection{A Generic Distributive Markov Category for the Graph Interface}
\label{sec:generic-dmc}
\newcommand{\CatG}{\mathbf{G}}
\newcommand{\GenericM}{\mathcal{T}}

We construct this generic category in two steps. We first create a distributive Markov category, actually a
distributive category, $\Fam(\CatG\op)$, that supports $(\tvertex,\tedge)$. We then add
$\tnew$ using the monoidal indeterminates method
of~\cite{hermida-tennent}.

\subsubsection{Step 1: A Distributive Category with $\tedge$}\label{sec:edge-cat}
\newcommand{\vertices}[1]{V_{#1}}
\newcommand{\edges}[1]{E_{#1}}

We first define a distributive category that supports
$(\tvertex,\tedge)$.
Let $\CatG$ be the category of finite graphs and functions that
preserve and reflect the edge relation. That is, a morphism $f:g\to
g'$ is a function $f:\vertices g \to \vertices g'$ such that for all
$v,w\in \vertices g$ we have $\edges g(v,w)$ if and only if $\edges
g(f(v),f(w))$. 

Recall (e.g.~\cite{hu-tholen}) that the free finite coproduct completion of a category
$\CatA$, $\Fam(\CatA)$ is given as follows. The objects of
$\Fam(\CatA)$ are sequences
$(X_1\dots X_n)$ of objects of $\CatA$, and the morphisms
$(X_1\dots X_m)\to (Y_1\dots Y_n)$ are pairs $(f,\{f_i\}_{i=1}^m)$
of a function $f:m\to n$ and a sequence of morphisms ${f_1:X_{1}\to Y_{f(1)}},\dots, {f_m:X_{m} \to
Y_{f(m)}}$ in $\CatA$.  

We consider the category $\Fam(\CatG\op)$. Let $\sem\tvertex=(1)$, the
singleton sequence comprising the one-vertex graph. 
\begin{proposition}\label{prop:FamG}
  \begin{enumerate}
  \item The free coproduct completion $\Fam(\CatG\op)$ is a distributive
    category, with the product $\sem\tvertex^n$ being the sequence of all
    graphs with $n$ vertices.
    In particular, $\sem\tvertex^2$ is a sequence with two components,
    the complete graph and the edgeless graph with two vertices. 
  \item Let $\tedge: \sem\tvertex\times \sem\tvertex\to 1+1$
    be the morphism $(\id,\{!,!\})$, intuitively returning true for
    the edge, and false for the edgeless graph. Here the terminal object $1$
    of $\Fam(\CatG\op)$ is the singleton tuple of the empty graph.
    This interpretation satisfies \eqref{eqn:simplegraph}. 
  \end{enumerate}
\end{proposition}
\begin{proof}[Proof notes.]
  Item~(1) follows from \cite{hu-tholen}, which shows that
  limits in $\Fam(\CatG\op)$ amount to ``multi-colimits'' in $\CatG$. 
  For example, the family of all graphs with $n$ vertices is a
  multi-coproduct of the one-vertex graph in $\CatG$,
  hence forms a product in $\Fam(\CatG\op)$.
  Item~(2) is then a quick calculation. All morphisms in $\Fam(\CatG\op)$ are deterministic. 
\end{proof}

\subsubsection{Step 2: Adjoining $\tnew$}
\newcommand{\CatHT}{\Fam(\CatG\op)[\nu]}

In Section~\ref{sec:edge-cat}, we introduced a distributive category
that interprets the interface $(\tvertex,\tedge)$.
But it does not support $\tnew$, and indeed there are no morphisms
$1\to\sem\tvertex$. 
To additionally
interpret $(\tnew)$, we freely adjoin it. We essentially use the `monoidal
indeterminates' method of Hermida and
Tennent~\cite{hermida-tennent} to do this.
Their work was motivated by semantics of dynamic memory allocation, but has also been
related to quantum phenomena~\cite{hs-quantum,ahk-partial-quantum} and to categorical gradient/probabilistic methods~\cite{para,backprop,Shiebler2021categorical}, where it is known as the `para construction'. It is connected to earlier methods for the action calculus~\cite{pavlovic_1997}.

Let $\Inj$ be the category of finite sets and injections. It is a
monoidal category with the disjoint union monoidal structure (e.g.~\cite{fiore-comb,power-local}).
Consider the functor $J:\Inj\op\to \Fam(\CatG\op)$, with
$J(n)=\sem\tvertex^n$, and where the functorial action is by
exchange and projection. This is a strong monoidal functor. (Indeed, it is
the unique monoidal functor with $J(1)=\sem\tvertex$.)

For any monoidal functor, Hermida and Tennent~\cite{hermida-tennent}
provide monoidal indeterminates by introducing a
`polynomial category', by analogy with a polynomial ring.
Unfortunately, a  general version for \emph{distributive} monoidal categories is not yet known,
so we focus on the specific case of $J:\Inj\op\to \Fam(\CatG\op)$. We build a new category
$\Fam(\CatG\op)[\nu:J\,\Inj\op]$, which we abbreviate $\CatHT$. It has the same objects as
$\Fam(\CatG\op)$, but the morphisms $\vec X\to \vec Y$ are equivalence classes of morphisms
\[
[k,f]: \sem\tvertex^{k}\times \vec X\to \vec Y
\]
in $\Fam(\CatG\op)$, modulo reindexing. The reindexing equivalence relation is generated by putting $[k,f]\sim [l,g]$ when
there exist injections $\iota_1\dots \iota_m:k\to l $ such that 
\[  g\  =\  \textstyle\Big(\sem\tvertex^{l}\times \vec X
  \cong \sum_{j=1}^m \sem\tvertex^l\times X_j
  \xrightarrow{\sem\tvertex^{(\iota_j)}\times X_j}
  \sum_{j=1}^m \sem\tvertex^k\times X_j\cong \sem\tvertex^k\times \vec X
  \xrightarrow f
  \vec Y\Big)
\]
where $\vec X=(X_1,\dots, X_m)$. In particular, when $m=1$, i.e. $\vec X=X$ is a singleton sequence, we have
\begin{equation}\label{eqn:CatHT}
  \CatHT(X,\vec Y)\  \cong\ \colim_{k\in\Inj}\Fam(\CatG\op)(\sem\tvertex ^k\times X,\vec Y)\text.
\end{equation}

Composition and monoidal structure accumulate in $\sem\tvertex^k$,
as usual in the monoidal indeterminates (`para') construction, 
e.g.
\[
  \Big(\vec X\xrightarrow {[k,f]} \vec Y
  \xrightarrow{[l,g]}
  \vec Z\Big)
  =
  \Big(\vec X\xrightarrow {[l+k,g\circ (\sem \tvertex^l \times f)]}
  \vec Z\Big)
\]
and although our equivalence relation is slightly coarser, it still respects the symmetric monoidal category structure,
and there is a monoidal functor $\Fam(\CatG\op)\to \CatHT$,
regarding each morphism $f:\vec X\to \vec Y$ in $\Fam(\CatG\op)$ as a morphism $[0,f]$ in $\CatHT$.
But there is also now an adjoined morphism $\nu=[1,\id]:1\to \sem\tvertex$.

This monoidal category $\CatHT$ moreover inherits the distributive coproduct structure from $\Fam(\CatG\op)$, and the functor $\Fam(\CatG\op)\to \CatHT$ is a distributive Markov functor.
To define copairing of $[k,f]:\vec X\to\vec Z$ and $[l,g]:\vec Y\to \vec Z$ we use the reindexing equivalence relation to assume $k=l$
and then define the copairing as $\langle [k,f],[k,g]\rangle=[k,\langle f,g\rangle] : \vec X+\vec Y\to\vec Z$. 

In summary:
  \begin{itemize}\item $\CatHT$ is a distributive Markov category.
  \item $\CatHT$ supports the graph interface, via
    the interpretation of $(\tvertex,\tedge)$ in $\Fam(\CatG\op)$, but also with the 
    interpretation
    $\sem{\tnew}=\nu:1\to \sem\tvertex$.
  \end{itemize}

\subsection{Bernoulli Bases for Random Graph Models}
\label{sec:bernoulli-base-graphon}
The following gives a precise characterization of graphons in terms of the numerals of $\CatHT$.

\begin{theorem}\label{thm:Markov-functor-graphon}
  To give a distributive Markov functor $\CatHT_\NN \to \FinStoch$ is to give a graphon.
\end{theorem}
\begin{proof}[Proof outline]
  We begin by showing a related characterization: that graphons correspond to certain natural transformations. 
  Observe that any distributive Markov category $\CatA$
  gives rise to a symmetric monoidal functor $\CatA(1,-):\FinSet_\NN\to\Set$, regarding the numerals of $\FinSet_\NN$ as objects of $\CatA$ (\S\ref{sec:bernoulli-base}). 
  Let $G_k=2^{k(k-1)/2}$ be the set of $k$-vertex graphs. 
  We can characterize the natural transformations $\alpha :
  \CatHT(1,-) \to \FinStoch(1,-)$ as follows.
  \begin{align*}
    &     \Nat(\CatHT(1,-)\ ,\ \FinStoch(1,-))\\
    \cong\ &\Nat\big(\colim_{k\in\Inj}\FinSet(G_k,-)\ ,\ \DistM(-)\big)
    &&\text{(Ex.~\ref{ex:bernoulli-base}(2), Prop.~\ref{prop:FamG}(1) and \eqref{eqn:CatHT})}\\
    \cong\ &\lim_{k\in \Inj\op}\Nat(\FinSet(G_k,-)\ ,\ \DistM(-))&&\text{(universal property of colimits)}\\
    \cong\ &\lim_{k\in \Inj\op}\DistM(G_k)&&\text{(Yoneda lemma)}
  \end{align*}
  An element of this limit of sets is by definition a sequence of distributions $p_k$ on $G_k$ that is invariant under reindexing by $\Inj\op$. Since injections are generated by inclusions and permutations, this is then a sequence that is consistent and exchangeable (Def.~\ref{def:rgm}), respectively.
  Such a natural transformation~$\alpha$ is monoidal if and only if
  the sequence is also \emph{local}. Hence a monoidal natural transformation is the
  same thing as a random graph model.

  In fact, every monoidal natural transformation $\alpha : \CatHT(1,-)
  \to \FinStoch(1,-)$ arises uniquely by restricting a distributive
  Markov functor $F : \CatHT_\NN \to \FinStoch$. We now show this, to
  conclude our proof. 
  Given $\alpha$, let $F_{m,n} : \CatHT_\NN(m,n) \to \FinStoch(m,n)$ be:
  \begin{displaymath}
    \CatHT_\NN(m,n) \cong \CatHT_\NN(1,n)^m \xrightarrow{\alpha_n^m} \FinStoch(1,n)^m \cong \FinStoch(m,n).
  \end{displaymath}
  It is immediate that this $F$ preserves the symmetric monoidal structure and coproduct structure, but not that $F$ is a functor.
  However, the naturality of $\alpha$ in $\FinSet_\NN$ gives us that $F$ preserves postcomposition by morphisms of $\FinSet_\NN$.
  All of this implies that general categorical composition is preserved as well, since, in any distributive Markov category of the form $\CatA_\NN$, for $f : l \to m$ and $g : m \to n$, the composite $g \circ f : l \to n$ is equal to
  \begin{displaymath}
    l = l \otimes 1^{\otimes m}
    \xrightarrow{f \otimes g_1 \otimes \ldots \otimes g_m}
    m \otimes n^{\otimes m} \xrightarrow{\text{eval}} n
  \end{displaymath}
  where $g_i = g \circ \iota_i$ for $i = 1,\ldots,m$ and $\text{eval}$ is just the evaluation map $m \times n^m \to n$ in $\FinSet$.
\end{proof}

\begin{corollary}\label{corollary:graphon-markov}
  Every graphon arises from a distributive Markov category via the random graph model in~\eqref{eqn:programgraph}.
\end{corollary}
\begin{proof}[Proof summary]
  Given a graphon, we consider the distributive Markov functor that
  corresponds to it, $\bbase : \CatHT_\NN \to \FinStoch$, by Theorem~\ref{thm:Markov-functor-graphon}.
  Using the quotient construction of Proposition~\ref{prop:quotient}, we get a distributive Markov category with a Bernoulli base.
  It is straightforward to verify that the random graph model induced by~\eqref{eqn:programgraph} is the original graphon.
\end{proof}

\subsection{Remark on Operational Semantics}
\label{sec:abstract-op-sem}
The interpretation in this section suggests a general purpose operational semantics for closed programs at ground type,
$\vdash t:n$, along the following lines:
\begin{enumerate}
\item Calculate the interpretation $\sem t:1\to n$ in $\CatHT$. There are no probabilistic choices in this step, it is a symbolic manipulation, because the morphisms of the Markov category
  $\CatHT$ are built from tuples of finite graph homomorphisms. In effect, this interpretation pulls all the $\tnew$'s to the front of the term.
\item Apply the Markov functor $\bbase(\sem t)$ to obtain a probability distribution on $n$, and sample from this distribution to return a result.
\end{enumerate}


%% file: random_free.tex
\section{Interpretation: black-and-white graphons via measure-theoretic probability}
\label{sec:randomfree}
In Section~\ref{sec:graphons-to-equational-theories}, we gave a general
syntactic construction for building an equational theory from a
graphon. Since that definition is based on free constructions and quotients, \emph{a priori}, it does not `explain' what the type $\tvertex$ stands for. Like contextual equivalence of programs, \emph{a priori}, it does not give useful compositional reasoning methods. To prove two programs are equal, according to the construction of Prop.~\ref{prop:quotient}, one needs to quantify over all $Z$, $h$, and $k$, in general.

In this section, we show that one class of graphons, black-and-white
graphons (Def.~\ref{def:bw}), admits a straightforward measure-theoretic semantics, and we can thus 
use the equational theory induced by this semantics, rather than the 
method of Section~\ref{sec:graphons-to-equational-theories}. This measure-theoretic semantics is close to 
previous measure-theoretic work on probabilistic programming
languages (e.g.~\cite{Kozen81,s-finite}). 

After recapping measure-theoretic probability (\S\ref{sec:probthy}), in Section~\ref{sec:bw-theory}, we show that every
black-and-white graphon arises from a measure-theoretic
interpretation (Prop.~\ref{prop:bw-ok}). In Section~\ref{sec:all-bw}, by defining `measure-theoretic interpretation' more
generally, we show that, conversely, this measure-theoretic approach can \emph{only} cater
for black-and-white graphons (Prop.~\ref{prop:all-bw}).

\subsection{Black-and-White Graphons from Equational Theories}\label{sec:bw-theory}
\begin{definition}\label{def:bw}[e.g.~\cite{janson}]
  A graphon $W:[0,1]^2\to [0,1]$ is \emph{black-and-white} if there exists
$E:[0,1]^2\to \{0,1\}$ such that 
$W(x,y)=E(x,y)$ for almost all $x,y$.
\end{definition}
Recall that the Giry monad (Def.~\ref{def:giry}) gives rise to a Bernoulli-based distributive Markov category  (\S\ref{sec:probthy}, Ex.~\ref{ex:bernoulli-base}).
For any black-and-white graphon~$W$, we define an interpretation of
the graph interface for the probabilistic programming language
using $\Giry$, as follows.
\begin{itemize}
\item $\sem\tvertex_W = [0,1]$; $\sem \tbool_W= 2$, the discrete two element space;
\item $\sem{\tnew()}_W = \mathrm{Uniform}(0,1)$, the uniform distribution on $[0,1]$;
\item $\sem{\tedge}_W(x,y)=\eta(E(x,y)).
        $
\end{itemize}
\begin{proposition}\label{prop:bw-ok}
  Let $W$ be a black-and-white graphon. The equational theory induced by $\sem-_W$ induces the graphon~$W$ according to the construction in Section~\ref{sec:program-equivalence-graphons}.
\end{proposition}
\begin{proof}
  

Suppose that $W$ corresponds to the sequence of random graphs $p_1,p_2,\dots $ as in Section~\ref{sec:graphons}. Consider the term $t_n$ in \eqref{eqn:programgraph}, and directly calculate its interpretation. Then, we get $\sem {t_n}_W=p_n$, via~\eqref{eqn:pwg}, as required. 

  The choice of $E$ does not matter in the interpretation of these terms, because $W=E$ almost everywhere.
\end{proof}
\subsection{All Measure-Theoretic Interpretations are
  Black-and-White}\label{sec:all-bw}
Although the model in Section~\ref{sec:bw-theory} is fairly canonical, there are sometimes other enlightening interpretations using the Giry monad. These also correspond to black-and-white graphons.

For example, consider the geometric-graph example from Figure~\ref{fig:graph}. We interpret this using the Giry monad, putting
\begin{itemize}
\item $\sem\tvertex = S_2$, the sphere; $\sem \tbool= 2$;
\item $\sem{\tnew()} = \mathrm{Uniform}(S_2)$, the uniform distribution on the sphere;
\item $\sem{\tedge}(x,y)=\eta(d(x,y)<θ)$, i.e.~an edge if their
  distance is less than $\theta$. 
\end{itemize}
This will again induce a graphon, via~(\ref{eqn:programgraph}).
We briefly look at theories that arise in this more flexible way:
\begin{proposition}
  \label{prop:all-bw}
  Consider any interpretation of the graph interface in the Giry monad:
  a measurable space $\sem\tvertex$,
  a measurable set $\sem\tedge\subseteq \sem\tvertex^2$, 
  and a probability measure $\sem{\tnew()}$ on $\sem\tvertex$.
  The induced graphon is black-and-white.
\end{proposition}
\begin{proof}[Proof notes]
 If $\sem\tvertex$ is standard Borel, the randomization lemma~\cite[Lem.~3.22]{kallenberg-2010} gives
  a function $f:[0,1]\to \sem\tvertex$ that pushes the uniform distribution on $[0,1]$ onto the probability measure $\sem{\tnew()}$.
        We define a black-and-white graphon $W$ by $W(x,y)=1$ if
        $(f(x),f(y)) \in \sem\tedge$, and $W(x,y)=0$ otherwise. 
This graphon interpretation $\sem-_W$ gives the same sequence of graphs in~\eqref{eqn:programgraph}, just by reparameterizing the integrals.

  If $\sem\tvertex$ is not standard Borel, we note that there is an equivalent interpretation where it is, because there exists a measure-preserving map $\sem\tvertex\to \Omega$ to a standard Borel space $\Omega$
  and a measurable set $E\subseteq\Omega^2$ that pulls back to $\sem\tedge$, giving rise to the same graphon
  (e.g.~\cite[Lemma~7.3]{janson}).
\end{proof}

\subsubsection*{Discussion}
Proposition~\ref{prop:all-bw} demonstrates that this measure-theoretic
interpretation has limitations.
\begin{definition}\label{def:e-r-graphon}
  For $α\in(0,1)$, the \emph{\ErdosRenyi\ graphon}  $W_α:[0,1]^2\to
  [0,1]$ is given by $W_α(x,y)=α$.
\end{definition}
The Erd\H{o}s-R\'enyi graphons cannot arise from measure-theoretic
interpretations of the graph interface, because they are not 
black-and-white. In Section~\ref{sec:ER-Rado}, we give an alternative
interpretation for the \ErdosRenyi\ graphons.

The reader might be tempted to interpret an \ErdosRenyi\ graphon by
defining
\[\sem{\tedge}_{W_α}(x,y)
  \quad=\quad
  \tbernoulli(\alpha).
\]
However, this interpretation does not provide a model
for the basic equations of the language, because this $\sem\tedge$
is not deterministic, and
derivable equations such as
\eqref{eqn:intro:det}
will fail. Intuitively, once an edge has been sampled between two
given nodes, its presence (or absence) remains unchanged in the rest
of the program, \ie the edge is not resampled again, it is memoized
(see also~\cite{kaddar-staton-stoch-mem,roy2008stochastic}).

Although not all graphons are black-and-white, these are
still a widely
studied and useful class. They are often called
`random-free'.
For example, an alternative characterization is that 
the random graph model of
Prop.~\ref{prop:lovasz}
has subquadratic entropy function~\cite[\S 10.6]{janson}. 
              

%% file: rado_nominal_sets.tex
\section{Interpretation: \ErdosRenyi\ graphons via Rado-nominal sets}
\label{sec:ER-Rado}
In Section~\ref{sec:graphons-to-equational-theories}, we gave a general construction to show that every graphon arises from a Bernoulli-based equational theory.
In Section~\ref{sec:randomfree}, we gave a more concrete interpretation, based on measure-theory, for black-and-white graphons. 
We now consider the \ErdosRenyi\ graphons (Def.~\ref{def:e-r-graphon}), which are not black-and-white. 

Our interpretation is based on Rado-nominal sets. These are also studied elsewhere, but for different purposes (e.g.~\cite{lmcs:1157,bojancyk-place,homomorphisms-fodef}, \cite[\S1.9]{pitts-nom-book}).

Rado-nominal sets (\S\ref{sec:rado-def}) are sets that are equipped with an action of the automorphisms of the Rado graph, which is an infinite graph that contains every finite graph.
There is a particular Rado-nominal set $\VV$ of the vertices of the Rado graph. The type $\tvertex$ will be interpreted as $\VV$;
$\tedge$ is interpreted using the edge relation~$E$ on $\VV$.
The equational theory induced by this interpretation gives rise to the \ErdosRenyi\ graphons (Def.~\ref{def:e-r-graphon}).

Since Rado-nominal sets form a model of ZFA set theory (Prop.~\ref{prop:radonom-topos}), we revisit
probability theory internal to this setting. We consider internal probability measures on Rado-nominal sets (\S\ref{sec:rado-prob}), and we show that there are internal probability measures on $\VV$ that give rise to \ErdosRenyi\ graphons (\S\ref{sec:er-meas}). The key starting point here is that, internal to Rado-nominal sets, the only functions $\VV\to 2$ are the sets of vertices that are definable in the language of graphs (\S\ref{sec:powersets-and-definiable-sets}). 

We organize the probability measures (Def.~\ref{def:probability}) into
a probability monad on Rado-nominal sets (\S\ref{sec:nom-monad}), analogous to the Giry monad. 
Fubini does not routinely hold in this setting (\S\ref{sec:fubini}), but we use a standard technique to cut down to a commutative affine monad (\S\ref{sec:comm-monad}). This gives rise to a Bernoulli-based equational theory, and in fact, this theory corresponds to \ErdosRenyi\ graphons (via \eqref{eqn:programgraph}: Corollary~\ref{cor:e-r}).

\subsection{Definition and First Examples}\label{sec:rado-def}
The Rado graph $(\VV,E)$ (\hspace{1sp}\cite{ackermannWiderspruchsfreiheitAllgemeinenMengenlehre1937,radoUniversalGraphsUniversal1964}, also known as the `random graph'~\cite{erdosRandomGraphs1959}) is the unique graph, up to isomorphism, with a countably infinite set of vertices
that has the extension property:
if $A$, $B$ are disjoint finite subsets of $\VV$, then there is a
vertex $a \in \VV \setminus (A \cup B)$ with an edge to all the vertices in $A$ but none of the
vertices in $B$.

The Rado graph embeds every finite graph, which can be shown by using
the extension property inductively.

An automorphism of the Rado graph is a graph isomorphism $\VV\to
\VV$. 
The automorphisms of the Rado graph relate to isomorphisms between
finite graphs, as follows. First, if $A$ is a finite graph regarded as
a subset of $\VV$, then any automorphism $\sigma$ induces an
isomorphism of finite graphs $A\cong \sigma[A]$. Conversely,
if $f\colon A\cong B$ is an isomorphism of finite graphs, and we regard $A$
and $B$ as disjoint subsets of $\VV$, then there exists an
automorphism $\sigma$ of $\VV$ that restricts to $f$ (i.e.\ $f=\sigma|_A$).

We write $\AutRado$ for the group of automorphisms of $(\VV,E)$. 
(This has been extensively studied in model theory and descriptive set theory, e.g. \cite{MR2140630,MR3274785}.)

\begin{definition}
  A \emph{Rado-nominal set} is a set $X$ equipped with an action
  $\bullet :\AutRado\times X\to X$ (i.e.~$\id\bullet x=x$; $(\sigma_2\cdot
  \sigma_1) \bullet x= \sigma_2\bullet \sigma_1\bullet x$) such that
  every element has finite support.

  An element $x\in X$ is defined to have \emph{finite support} if there is a finite set
  $A\subseteq \VV$ such that for all automorphisms $\sigma$, if $\sigma$ 
  fixes $A$ (i.e.~$\sigma|_A=\id_A$), it also fixes $x$ (i.e.~$\sigma\bullet x=x$).

  \emph{Equivariant functions between Rado-nominal sets} are functions that
  preserve the group action (i.e.~$f(\sigma\bullet x)=\sigma\bullet (f(x))$). 
\end{definition}
\begin{proposition}[\hspace{1sp}\cite{pitts-nom-book}]
  If finite sets $A,B\subseteq \VV$ both support $x$, so does $A\cap B$. Hence every
  element has a least support. 
\end{proposition}
\begin{example}
  \begin{enumerate}
    \item The set $\VV$ of vertices is a Rado-nominal set, with
      $\sigma\bullet a=\sigma(a)$. The support of vertex $a$ is $\{a\}$.
    \item The set $\VV\times \VV$ of pairs of vertices is a
      Rado-nominal set, with
            $\sigma\bullet (a,b)=(\sigma(a),\sigma(b))$. The support
            of $(a,b)$ is $\{a,b\}$.
            More generally, a finite product of Rado-nominal sets has a
            coordinate-wise group action.
          \item The edge relation $E\subseteq \VV\times \VV$ is a
            Rado-nominal subset (which is formally defined in \S\ref{sec:powersets-and-definiable-sets}) because automorphisms preserve the edge relation. 
            \item Any set $X$ can be regarded with the discrete
              action,
              $\sigma\bullet x=x$, and then every element has empty
              support.
              We regard these sets with the discrete action:
              $1=\{\star\}$; $2=\{0,1\}$; $\NN$; and the unit interval $[0,1]$.
              \end{enumerate}
          \end{example}
\subsection{Powersets and Definable Sets}
\label{sec:powersets-and-definiable-sets}
For any subset $S\subseteq X$ of a Rado-nominal set, we can define
$\sigma\bullet S=\sigma[S]= \{\sigma\bullet x~|~x\in S\}$. 
We let
\begin{equation}2^X=\{S\subseteq X~|~S \text{ has finite support}\}\text.\label{eqn:powerobject}\end{equation}
This is a Rado-nominal set.

\begin{example}\label{example:subset}
  We give some concrete examples of subsets.
\begin{enumerate}
  \item For vertices $b$ and $c$ in $\VV$ with no edge between them,
    the
    set $\{a\in\VV~|~E(a,b) ∧ E(a,c)\}$ is the set of ways of
    forming a horn. It has support $\{b,c\}$. 
  \item $\{(b,c)\in \VV^2~|~E(a,b) ∧ E(a,c) ∧ \neg E(b,c)\}$ is the
    set of horns with apex $a$; it has support $\{a\}$.
  \item $\{(a,b,c)\in \VV^3~|~E(a,b) ∧ E(a,c) ∧ \neg E(b,c)\}$ is the
    set of all oriented horns; it has empty support.
  \item (Non-example) There is a countable totally disconnected subgraph of $\VV$; it does not have finite support as a subset of $\VV$.
  \end{enumerate}
  \end{example}
  In fact, the finitely supported subsets correspond exactly to the
  definable sets in first-order logic over the theory of graphs. The
  following results may be folklore.

  \begin{proposition}\label{prop:radonom-definable-sets}
    Let $S\subseteq \VV^n$, and $A\subseteq \VV$ be finite. The following are equivalent:
    \begin{itemize}
      \item $S=\{(s_1,\dots s_n)~|~\phi(s_1\dots s_n)\}$, for a
        first-order formula
        $\phi$ over the theory of graphs, with parameters in~$A$;
      \item $S$ has support~$A$.
        \end{itemize}
      \end{proposition}
      \begin{proof}
      $(\Rightarrow)$ For all isomorphisms $f\colon \VV \to \VV$ that fix $A$, and for all elements $a_1\dots a_k \in A$ and 
      subsets
      $S=\{(s_1,\dots, s_n)~|~\phi(s_1\dots s_n, a_1\dots a_k)\}$,
      we have 
\[       \phi(f(s_1)\dots f(s_n), a_1\dots a_k) \quad= \quad
       \phi(f(s_1)\dots f(s_n), f(a_1)\dots f(a_k)) \text .
      \]
      Furthermore, $\phi$ is invariant with respect to $f$. Thus, the image $f(S) \subseteq S$. By a similar argument, we have $f\inv(S) \subseteq S$, so that $S \subseteq f(S)$. Thus, $f(S) = S$
      (\hspace{1sp}\cite[Prop. 1.3.5]{markerModelTheoryIntroduction2002}).

      $(\Leftarrow)$
      This is a consequence of the Ryll-Nardzewski theorem for the theory of the Rado graph (which can be shown to be $ω$-categorical by a back-and-forth argument, using the extension property of the Rado graph). But we give here a more direct proof, assuming $n = 1$ for simplicity. Suppose $A \subseteq \VV$ is a finite support for $S$. Then, for any $v,v' \in \VV \backslash A$, if $v$ and $v'$ have the same connectivity to $A$, then they are either both in or not in $S$ since, by the extension property, we can find an automorphism fixing $A$ and sending $v$ to $v'$. The set of vertices with the same connectivity to $A$ as $v$ is definable, and there are only $2^{|A|}$ such sets. Hence, $S \backslash A$ is a union of finitely many definable sets, and as $S \cap A$ is definable (being finite), so is $S = (S \backslash A) \cup (S \cap A)$.
      \end{proof}

      We note that $2^X$ in \eqref{eqn:powerobject} is a canonical
      notion of internal powerset, from a categorical perspective.
\begin{proposition}\label{prop:radonom-topos}
$\RadoNom$ is a Boolean Grothendieck topos, with powerobject $2^X$ in \eqref{eqn:powerobject}. 
\end{proposition}
\begin{proof}[Proof notes]
  $\RadoNom$ can be regarded as continuous actions of $\AutRado$, regarded as a topological group with the product topology, and then we invoke standard methods~\cite[Ex.~A2.1.6]{elephant}. It is also equivalent to the category of sheaves over finite graphs and embeddings with the atomic topology. See~\cite{caramello-fraisse,caramello-toposic-galois} for general discussion.
\end{proof}

      \subsection{Probability Measures on Rado-Nominal Sets}\label{sec:rado-prob}
      The finitely supported sets $S\subseteq \VV$ can be regarded as
      `events' to which we would assign a probability. For example,
      if we already have vertices $b$ and $c$, we may want to know the
      chance
      of picking a vertex that forms a horn, and this would be
      the probability of the set in Ex.~\ref{example:subset}(a).
      \begin{definition}\label{def:probability}
        A sequence $S_1,S_2\dots \subseteq X$ is said to be
        \emph{support-bounded} if there is one finite set $A\subseteq\VV$ that
        supports all the sets $S_i$.

        A function $\mu:2^X\to[0,1]$ is \emph{(internally)
          countably additive} if for any
        support-bounded sequence $S_1,S_2\dots \subseteq X$ of disjoint sets, 
        \[\textstyle
          \mu(\biguplus_{i=1}^\infty S_i)=\sum_{i=1}^\infty\mu(S_i)\text.
        \]
        
        A \emph{probability measure} on a Rado-nominal set $X$ is an equivariant function
        $\mu:2^X\to [0,1]$ that is internally countably additive, such that $\mu(X) = 1$. \end{definition}
      We remark that there are two subtleties here. First,
      we restrict to support-bounded sequences. These are the
      correctly internalized notion of sequence in Rado-nominal
      sets,
      since they correspond precisely to finitely-supported functions $\NN\to
      2^X$.
      Second, we consider a Rado-nominal set to be equipped with its internal powerset $2^X$, 
      rather than considering 
      sub-$\sigma$-algebras.


%% file: graphons_to_equations.tex
\newcommand{\MonadT}{\mathcal{T}}
\newcommand{\MonadP}{\mathcal{P}}
\newcommand{\MonadF}{\mathcal{F}}

\paragraph{Measures on the space of vertices}\label{sec:er-meas}
We define an internal probability measure
(Def.~\ref{def:probability}) on the space~$\VV$ of
vertices, which, we will show, corresponds to the Erd\H{o}s-R\'enyi
graphon.
Fix $α\in[0,1]$, the chance of an edge. 

We define the measure $\nu_α$ of a definable set $S\in 2^\VV$ as follows.
Suppose that $S$ has support $\{a_1,\dots, a_n\}$.
We choose an enumeration of vertices $(v_1,\dots,v_{2^n})$ in $\VV$
(disjoint from $\{a_1,\dots,a_n\}$)
that covers all the $2^n$ possible edge relationships that a vertex
could have with the $a_i$'s. (For example, $v_1$ has no edges to any $a_i$, and $v_{2^n}$ has an edge to
every $a_i$, and the other $v_j$'s have the other possible edge relationships.) Let:
\begin{equation}\label{eqn:graphon-to-measure}
  \nu_α(S) =\sum_{j=1}^{2^n}[v_j\in S] \prod_{i=1}^n \big(α E(v_j,a_i)+(1-α)(1-E(v_j,a_i))\big)\text. 
\end{equation}
\begin{proposition}
  The assignment given in \eqref{eqn:graphon-to-measure}
  is an internal probability measure (Def.~\ref{def:probability}) on $\VV$. 
\end{proposition}
\begin{proof}
The function $\nu_α$ is well-defined: it does not depend on the choice of $v_j$'s (by Prop.~\ref{prop:radonom-definable-sets}), nor on the choice of support (by direct calculation).
It is equivariant, since for $\sigma \bullet S$, a valid enumeration of vertices is given by $\sigma \bullet v_1, \dots \sigma \bullet v_{2^{n}}$. Also, $\nu(\VV) = 1$, since $\VV$ has empty support. Internal countable additivity follows from the identity $
\left[v_j \in \biguplus_{i=1}^\infty S_i\right] = \sum_{i=1}^{\infty}[v_j \in S_i]
$.
\end{proof}

    \paragraph*{Remark}
The definitions and results of this section appear to be novel. 
    However, the general idea of considering measures on formulas which are invariant to substitutions that permute the variables
goes back to work of Gaifman \cite{MR175755}.
The paper~\cite{MR3515800}
characterizes those countably infinite graphs that can arise with probability~$1$ in that framework;
see \cite{properly-ergodic}
for a discussion of how Gaifman's work connects to
Prop.~\ref{prop:lovasz}.

\subsection{Nominal Probability Monads}
\label{sec:nom-monad}
Since $\RadoNom$ is a Boolean topos with natural numbers object
(Prop.~\ref{prop:radonom-topos}), we can interpret measure-theoretic
notions in the internal language of the topos, as long as they do not
require the axiom of choice. We now spell out the resulting development, without assuming familiarity with topos theory. By doing this, we build new probability
monads on $\RadoNom$. 

\subsubsection{Finitely Supported Functions and Measures}
Let $X$ and $Y$ be Rado-nominal sets.
The set of all functions $X\to Y$ has an action of $\AutRado$,
given by $(\sigma\bullet f)(x)=\sigma\inv\bullet (f(\sigma\bullet
x))$. The function space $[X\Rightarrow Y]$ comprises those
functions that have finite support under this action. Categorically,
this structure is uniquely determined by the `currying' bijection, natural in~$Z$:
\[
  \RadoNom(Z\times X,Y)\cong \RadoNom(Z,X\Rightarrow Y)\text .
\]
(For example, the powerobject $2^X$ (\S\ref{sec:powersets-and-definiable-sets}) can be regarded as
$[X\Rightarrow 2]$, if we regard a set as its characteristic
function.)

In Def.~\ref{def:probability}, we focused on equivariant
probability measures.
We generalize this to finitely supported measures. For example, pick a
vertex $a\in \VV$. Then, the Dirac measure on $\VV$ (i.e.~$\delta_a(S)=1$ if $a\in S$,
and $\delta_a(S)=0$ if $a\not\in S$) has support $\{a\}$.

\begin{definition}
\label{def:finitely-supported-probability-measure}
For a Rado-nominal set $X$, let $\MonadP (X)$ comprise the
    finitely supported functions $\mu : 2^X\to[0,1]$ that are
    internally countably additive, and satisfy $\mu(X) = 1$.
    This is a Rado-nominal set, as a subset of
    $[2^X\Rightarrow[0,1]]$. Functions in $\MonadP (X)$
    are called \emph{finitely supported probability measures}.
\end{definition}

\subsubsection{Internal Integration}
We revisit some basic integration theory in this nominal setting.
In traditional measure theory, one can define the Lebesgue integral of a \emph{measurable} function ${f : X \to [0,1]}$ by
$
  \int f(x) \mu(\dd x) = \sup \sum_{i = 1}^n r_i \mu(U_i)
$
where the supremum ranges over simple functions ${\sum_i r_i [- \in U_i]}$ with $U_i$ measurable in $X$ and bounded above by $f$ (\S\ref{sec:probthy}).
The same construction works in the internal logic of $\RadoNom$.

Note that the following does not mention $f$ being measurable: since $X$ is considered to have its internal powerset $\sigma$-algebra, finite-supportedness 
implies `measurability' here.

\begin{proposition}\label{prop:internal-integration}
  Let $\mu\in \MonadP (X)$ be a finitely supported probability measure
  on $X$.
  For any finitely supported function $f\colon X\to [0,1]$,
  the internally-constructed Lebesgue integral $\int f(x)\,\mu(\dd x) \in [0,1]$ exists.
  Moreover, integration is an equivariant map
  \[
    \int : \MonadP (X)\times [X\Rightarrow [0,1]]\to [0,1]
  \]
  which preserves suprema of internally countable monotone sequences in its second argument.
  \end{proposition}
  \begin{proof}
    If $U_1,\ldots,U_n \subseteq X$ are finitely supported, $r_1,\ldots,r_n \in [0,1]$, and $\sum_i r_i [- \in U_i] \leq f$, then by ordinary additivity of $\mu$, we have $\sum r_i \mu(U_i) \in [0,1]$.
    By ordinary real analysis, the supremum of all such values exists and is in $[0,1]$.
    For equivariance, recall that $[0,1]$ is equipped with the trivial action of $\AutRado$. Use the fact that $\sum_i r_i [- \in U_i] \leq f$ if and only if $\sum_i r_i [- \in \sigma\bullet U_i] \leq \sigma \bullet f$.
    The last claim is the monotone convergence theorem internalized to $\RadoNom$.
  \end{proof}

\subsubsection{Kernels and a Monad}
  We can regard a `probability kernel' as a finitely supported function
  $k:X\to \MonadP (Y)$. Equivalently, $k$ is a finitely supported function
  $k:X\times 2^Y\to[0,1]$ that is countably additive and has mass $1$ in its second
  argument. 

  (In traditional measure theory, one would explicitly ask that $k$ is measurable
  in its first argument, but as we observed, finite-supportedness already implies it.)
  
  As usual, probability kernels compose, and this allows us to regard
  them as Kleisli morphisms for a monad (Def.~\ref{def:monad}), defined as follows.
\begin{definition}\label{def:monad-radonom}
  We define the strong monad $\MonadP$ on $\RadoNom$ as follows.
  \begin{itemize}
  \item For a Rado-nominal set $X$, let $\MonadP (X)$ comprise the
    finitely supported probability measures (Def.~\ref{def:finitely-supported-probability-measure}).
  \item
    The unit of the monad $\eta_X:X\to \MonadP (X)$ is the Dirac measure,
    $\eta_X(x)(S)=[x\in S]$.
  \item
    The bind $(\bind) \colon \MonadP (X)\times (X\Rightarrow \MonadP (Y))\to \MonadP (Y)$ is given
    by
    \[
      (\mu \bind k)(S)=\int_X k(x,S)\,\mu(\dd x).
    \]
  \end{itemize}
\end{definition}
We note that this is similar to the `expectations monad'~\cite[Thm.~4]{EPTCS95.12}. 
\subsubsection{Commuting Integrals (Fubini)}
  \label{sec:fubini}
  For measures $\mu_1\in \MonadP (X)$ and $\mu_2\in \MonadP (Y)$, the
  monad structure allows us to define a product measure
  \begin{equation}
    \begin{aligned}
    & \mu_1\otimes\mu_2 = \big(\mu_1\bind (\lambda x.\,\mu_2\bind\,\lambda
    y.\,\eta(x,y))\big)
    \\
    & \int f(x,y)\,(\mu_1\otimes \mu_2)(\dd (x,y)) =
    \int\int f(x,y)\,\mu_2(\dd y)\,\mu_1(\dd x)\text .
  \end{aligned}
\end{equation}
Although this iterated integration is reminiscent of the traditional
approach, in general we cannot reorder integrals (`Fubini does not
hold'). For example, given two measures $\nu_α$ and $\nu_β$ for $α \neq β$ and
$f$ being the characteristic function of the set $\{(x, y): E(x, y)\} \subseteq \VV^2$, we have
\begin{equation}
\begin{aligned}
& \int \int [E(x, y)]\,\nu_α(\dd y)\,\nu_β(\dd x) = \int α\, \nu_β(\dd x) = α  \\
&\;\; {} \neq β = \int β\, \nu_α(\dd y) = \int \int [E(x, y)]\,\nu_β(\dd x)\,\nu_α(\dd y) \text .
\end{aligned}
\end{equation}
However, it does hold when we consider only copies of the same measure.
\begin{proposition}\label{thm:graphon:commutative}
  For $\nu_α\in \MonadP (\VV)$ as in (\ref{eqn:graphon-to-measure}), $\nu_α$ commutes with
  $\nu_α$. That is, for any finitely supported $f:\VV\times\VV\to [0,1]$,
  \[     \int \int f(x,y)\,\nu_α(\dd y)\,\nu_α(\dd x)
    =
       \int\int f(x,y)\,\nu_α(\dd x)\,\nu_α(\dd y)
    \text.\]
\end{proposition}
\begin{proof}[Proof notes]
  By Prop.~\ref{prop:radonom-definable-sets} and \ref{prop:internal-integration}, it suffices to check on the indicator functions of definable subsets of $\VV^2$.
  The indicators of sets $\{(x,y) \mid \Phi(x,y)\}$ where $\Phi(x,y)$ is a disjunction of $x = y$, $x = a$, or $y = a$ for some $a \in \VV$ are seen to have integral $0$ on both sides.
  The remaining possibilities can be reduced to 
  the case where $\Phi_{A,\phi,\psi,\epsilon}(x,y)$ is $(x,y \notin A) \wedge (x \neq y) \wedge (E(x,y) \leftrightarrow \epsilon) \wedge \bigwedge_{a \in A} (E(a,x) \leftrightarrow \phi_a) \wedge (E(a,y) \leftrightarrow \psi_a)$
  where $A \subseteq \VV$ is a finite set, $\epsilon \in \{\bot,\top\}$, and $\phi,\psi \in \{\bot,\top\}^A$.
  This formula corresponds to choosing a two-vertex extension of the finite graph spanned by $A \subseteq \VV$.
  Intuitively, the two double integrals correspond to the two alternative two-step computations of the conditional probability of extending the graph $A$ to this extension according to which of the two vertices is sampled first, and indeed both evaluate to $α^k (1-α)^{2|A|+1 - k}$ where $k = [\epsilon] + \sum_{a \in A} ([\phi_a] + [\psi_a])$.
\end{proof}

\paragraph*{Remark} In traditional measure theory, iterated integrals are defined using product $\sigma$-algebras. Here we have not constructed product $\sigma$-algebras, but rather always take the internal powerset as the $\sigma$-algebra. This allows us to view all the definable sets as measurable on $\VV^n$ (Prop.~\ref{prop:radonom-definable-sets}), which is very useful. We remark that alternative product spaces also arise in non-standard approaches to graphons (see~\cite[\S6]{tao-graphons} for an overview), and also in quasi-Borel spaces~\cite{qbs} for different reasons.

\subsubsection{A Commutative Monad}
  \label{sec:comm-monad}
  We now use Prop.~\ref{thm:graphon:commutative} to build a commutative affine submonad $\MonadP_α$ of the
  monad~$\MonadP$,
  which we will use to model the graph interface for the probabilistic programming language.
  With Prop.~\ref{prop:radonom-topos}, we use the following general result.
  \begin{proposition}
    Let $\MonadT$ be a strong monad on a Grothendieck topos. Consider a family of morphisms
    $\{f_i\colon X_i\to \MonadT (Y_i)\}_{i\in I}$.
    \begin{itemize}
    \item There is a least strong submonad $\MonadT_{f}\subseteq \MonadT$ through which all
      $f_i$ factor.
    \item If the morphisms $f_i$ all commute with each other, then
      $\MonadT_f$ is a commutative monad (Def.~\ref{def:affine-monad}).
    \end{itemize}
  \end{proposition}
  \begin{proof}[Proof notes]\newcommand{\Sub}{\mathrm{Sub}}
    Our argument is close to \cite[\S2.3]{kammar-mcdermott} and also
    \cite[Thms.~7.5 \& 12.8]{kammar}.

    We let $\MonadT_f$ be the least subfunctor of $\MonadT$ that contains the images of the $f_i$'s and $\eta$, and is closed under the image of monadic bind ($\bind$).
    To show that this exists, we proceed as follows. First, fix a regular cardinal $\lambda>I$ such that $Y_i$'s are all $\lambda$-presentable, such that the topos is locally $\lambda$-presentable (e.g.~\cite{lp-adamek-rosicky}). 
    Consider the poset $\Sub_\lambda(\MonadT)$ of $\lambda$-accessible subfunctors of $\MonadT$. The cardinality bound $\lambda$ ensures it is small. Ordered by pointwise inclusion, this is a complete lattice: the non-empty meets are immediate, and the empty meet requires us to consider the $\lambda$-accessible coreflection of $\MonadT$.

We defined $\MonadT_f$ by a monotone property which we can regard as a monotone operator on this complete lattice $\Sub_\lambda(\MonadT)$, and so the least $\lambda$-accessible subfunctor exists. This is $\MonadT_f$. Concretely, it is a least upper bound of an ordinal indexed chain.
    The chain starts with the functor
    \[\textstyle F_0(Z)=\bigcup_{i\in I,g:Y_i\to Z}\mathsf{image}(\MonadT(g)\circ f_i)\quad\subseteq
      \MonadT(Z)\]
    which is $\lambda$-accessible because the $Y_i$'s are $\lambda$-presentable.
    The chain iteratively closes under the image of monadic bind, until we reach a subfunctor that is a sub\emph{monad} of $\MonadT$. 
    
    To see that $\MonadT_f$ is commutative, we appeal to (transfinite) induction. Say that a subfunctor~$F$ of $\MonadT$ is \emph{commutative} if all morphisms that factor through $F$ commute (Def.~\ref{def:affine-monad}), and then note that the property of being commutative is preserved along the ordinal indexed chain. 
        \end{proof}
  With this in mind, fixing a measure
  $\nu_α$ as in \eqref{eqn:graphon-to-measure},
  we form the least submonad $\MonadP_α$ of $\MonadP$ induced by the morphisms
  \begin{equation} \label{eqn:generators-monad}
    \nu_α:1\to \MonadP (\VV)\qquad
   \tbernoulli:[0,1]\to \MonadP(2) 
  \end{equation}
  where $\tbernoulli(r)=r\cdot \eta(0)+(1-r)\cdot \eta(1)$.
  \begin{corollary}\label{coro:least-submonad-is-comm-affine}
    The least submonad $\MonadP_α$ of the probability monad $\MonadP$ induced by
    the morphisms in \eqref{eqn:generators-monad}
    is a commutative affine monad (Def.~\ref{def:affine-monad}).
  \end{corollary}
\begin{proof}[Proof notes]  It is easy to show that $\tbernoulli$ commutes with every morphism
  $X\to \MonadP (Y)$. Moreover, $\nu_α$ commutes with
  itself (Prop.~\ref{thm:graphon:commutative}).
Finally, $\MonadP_α$ is affine since $\MonadP$ is. \end{proof}
\subsection{Summary and Interpretation}
Fix $\alpha\in [0,1]$. We induce an internal measure~$\nu_α$ on the vertices of the Rado
    graph as explained in \eqref{eqn:graphon-to-measure};
    and build a commutative submonad $\MonadP_α$ of $\MonadP$.
    We can then interpret the graph probabilistic programming
    language.
    We interpret types as Rado-nominal sets:
\begin{equation}\label{eqn:radointerp}    
            \sem{\tbool}  = 2
            \qquad 
            \sem{\tvertex}  = \VV
            \qquad
            \sem{\tunit}  = 1
            \qquad
            \sem{A_1\ast A_2}  = \sem{A_1}\times \sem{A_2}\text.
\end{equation}
    We interpret typed programs $\Gamma\vdash t:A$ as Kleisli morphisms
    \[
      \sem\Gamma\to \MonadP_α(\sem A)
    \]
    i.e. internal probability kernels $\sem\Gamma \times 2^{\sem A}\to [0,1]$. 
    Sequencing (let) is interpreted using the monad structure, with 
    $\sem\tnew:1\to \MonadP_\alpha(\VV)$ and
    $\sem\tedge:\VV\times \VV\to \MonadP_\alpha(2)$ as
    \begin{equation}\label{eqn:rado-interp-b}
      \sem{\tnew()} = \nu_α \qquad
      \sem{\tedge}(v,w) = \eta(E(v,w)) \end{equation}
    \begin{corollary}\label{cor:e-r}
      Consider the interpretation in Rado-nominal sets (\eqref{eqn:radointerp}--~\eqref{eqn:rado-interp-b}).
         If we form the sequence of random graphs
          in \eqref{eqn:programgraph},
          then these correspond to the Erd\H{o}s-R\'enyi graphon.
      \end{corollary}
      \begin{proof}[Proof notes.]
        The semantics interprets ground types as finite sets with discrete $\AutRado$ action -- in which case internal probability kernels correspond to stochastic matrices, agreeing with $\FinStoch$. Thus, the theory is Bernoulli-based.
      To see that the graphon arises, consider for instance when $n=2$, we have:
      \[
          \sem{t_2}(\star)= 
           \int \begin{pmatrix}
              [E(x_1,x_1)], [E(x_1,x_2)]
              \\
              [E(x_2,x_1)], [E(x_2,x_2)]
              \end{pmatrix} (\nu_α \otimes \nu_α)(\dd(x_1, x_2))
      \]
      for $t_2$ as in \eqref{eqn:programgraph}, and therefore
      \[
      \sem{t_2} = 
        \begin{pmatrix}
          \delta_0, &\tbernoulli(α)
          \\
          \tbernoulli(α), &\delta_0
        \end{pmatrix}:\MonadP(2^4)
      \]
      For general $n$, this corresponds to the random graph model $p_{W_α,n}$ for the Erd\H{o}s-R\'enyi graphon~$W_α$.
    \end{proof}


%% file: related_work.tex
\section{Conclusion}
\label{sec:relatedwork}
\paragraph{Summary}
We have shown that equational theories for the graph interface to the probabilistic programming language (Ex.~\ref{ex:interfaces}) give rise to graphons (Theorem~\ref{thm:eq-theory-to-graphon}).
Conversely, every graphon arises in this way. We showed this generally using an abstract construction based on Markov categories (Corollary~\ref{corollary:graphon-markov}) and methods from category theory~\cite{hermida-tennent,hu-tholen}. Since this is an abstract method, we also considered two concrete styles of semantic interpretation that give rise to classes of graphons: traditional measure-theoretic interpretations give rise to black-and-white graphons (Prop.~\ref{prop:bw-ok}), and an interpretation using the internal probability theory of Rado-nominal sets gives rise to \ErdosRenyi\ graphons (Corollary~\ref{cor:e-r}).

\paragraph{Further context, and future work}

The idea of studying exchangeable structures through program equations is perhaps first discussed in the abstract~\cite{XRPDA-PPS2017}, whose \S3.2 ends with an open question about semantics of languages with graphs that the present paper addresses.
Subsequent work addressed the simpler setting of exchangeable sequences and beta-bernoulli conjugacy through program equations~\cite{Staton-BB18-short},
and stochastic memoization~\cite{kaddar-staton-stoch-mem}; the latter uses a category similar to $\RadoNom$, although the monad is different. 
Beyond sequences~\cite{Staton-BB18-short} and graphs (this paper), a natural question is how to generalize to arbitrary exchangeable interfaces (see e.g.~\cite{6847223}). 
For example, we could consider exchangeable random boolean arrays via the interface
\[
  \mathsf{new\text-row}:\tunit\to\mathsf{row},\quad\mathsf{new\text-column()}:\tunit\to\mathsf{column},\quad
  \mathsf{entry}:\mathsf{row}*\mathsf{column}\to \tbool
\]
and random hypergraphs with the interface 
\[
  \mathsf{new}:\tunit \to \mathsf{vertex},\quad
  \mathsf{hyperedge}_n:\mathsf{vertex}^n\to \tbool\text.
\]
We could also consider interfaces for hierarchical structures, such as arrays where every entry contains a graph. Diverse exchangeable random structures have been considered from the model-theoretic viewpoint~\cite{ackerman-autm-invariant,crane-towsner} and from the perspective of probability theory (e.g.~\cite{kallenberg-2010,Campbell23,jlsy-dags}), but it remains to be seen whether the programming perspective here can provide a unifying view. Another point is that graphons correspond to dense graphs, and so a question is how to accommodate sparse graphs from a programming perspective (e.g.~\cite{10.1111/rssb.12233,10.1214/18-AOS1778}). 

This paper has focused on a very simple programming language~(\S\ref{sec:generic-ppl}).
As mentioned in Section~\ref{sec:intro:practice},
several implementations of probabilistic programming languages do support various Bayesian nonparametric primitives based on exchangeable sequences, partitions, and relations
(e.g. \cite{wood-aistats-2014,goodman2008church,roy2008stochastic,ks-prob-first-class-store,lazyppl,msp-venture}).
In particular, the `exchangeable random primitive' (XRP) interface
\cite{XRP-PPS2016,wu-church} 
provides a built-in abstract data type for representing exchangeable sequences.
This aids model design by its abstraction, but also aids inference performance by clarifying the independence relationships. 

Aside from practical inference performance, we can ask whether representation and inference are computable.
For the simpler setting of exchangeable sequences, this is dealt with positively by  \cite{freer2012computable, pmlr-v9-freer10a}.
The question of computability for graphons and exchangeable graphs is considerably subtler,
and some standard representations are noncomputable
\cite{ackerman2019algorithmic} (see also \cite{CompExch-PPS2017}).
This suggests several natural questions about whether certain natural classes of computable exchangeable graphs
can be identified by program analyses in the present context.
